\documentclass[sigconf,
			final
			nonacm
                        ]{acmart}
\renewcommand\footnotetextcopyrightpermission[1]{} 
\newcommand{\takeout}[1]{\empty}

\usepackage[utf8]{inputenc}
\usepackage[T1]{fontenc}

\usepackage{graphicx}
\usepackage{float}       

\usepackage[capitalise,nameinlink]{cleveref}
\usepackage[utf8]{inputenc}
\usepackage{latexsym}
\usepackage{mathrsfs}
\usepackage{dsfont}

\usepackage{tikz-cd}
\usetikzlibrary{arrows,arrows.meta,calc}
\tikzcdset{arrow style=tikz, diagrams={>=to}}

\tikzstyle{shiftarr}=[
        rounded corners,%
        to path={--([#1]\tikztostart.center)
                     -- ([#1]\tikztotarget.center) \tikztonodes
                     -- (\tikztotarget)},
]

\newcommand{\pullbackangle}[2][]{\arrow[phantom,to path={
                     -- ($ (\tikztostart)!1cm!#2:([xshift=8cm]\tikztostart) $)
                        node[anchor=west,pos=0.0,rotate=#2,
                        inner xsep = 0]
                        {\begin{tikzpicture}[minimum
                        height=1mm,baseline=0,#1]
    \draw[-] (0,0) -- (.5em,.5em) -- (0,1em);
                        \end{tikzpicture}}}]{}}

\usepackage{framed}
\usepackage{multicol}
\usepackage{mathtools}
\usepackage{soul}

\newcommand{\resetCurThmBraces}{%
\gdef\curThmBraceOpen{(}%
\gdef\curThmBraceClose{)}}
\resetCurThmBraces
\newcommand{\removeThmBraces}{%
\gdef\curThmBraceOpen{}%
\gdef\curThmBraceClose{}}

\newenvironment{notheorembrackets}{\removeThmBraces}{\resetCurThmBraces}
\usepackage{etoolbox}
\patchcmd{\thmhead}{(#3)}{\curThmBraceOpen #3\curThmBraceClose }{}{}

\takeout{
\declaretheoremstyle[
    spaceabove=\ownthmSpaceAbove,
    spacebelow=\ownthmSpaceBelow,
    headpunct=.,
    postheadspace=.5em,
    notebraces={\curThmBraceOpen}{\curThmBraceClose},
    postheadhook={\resetCurThmBraces},
]{definition}
\declaretheoremstyle[
    style=definition,
    bodyfont=\itshape,
    notebraces={\curThmBraceOpen}{\curThmBraceClose},
    postheadhook={\resetCurThmBraces},
]{theorem}}

\numberwithin{equation}{section}

 \usepackage[author=anonymous,nomargin,marginclue,footnote,final]{fixme}
 \FXRegisterAuthor{ls}{als}{LS}
 \FXRegisterAuthor{cf}{acf}{CF}
 \FXRegisterAuthor{sm}{asm}{SM}
 \FXRegisterAuthor{bk}{abk}{BK}
 \FXRegisterAuthor{hb}{ahb}{HB}
\usepackage{siunitx}
\usepackage{hyperref}
\hypersetup{hidelinks,final,hypertexnames=false}

\usepackage{amsmath}
\usepackage{amsthm}   

\usepackage[inline]{enumitem}
\setlist[enumerate,1]{label=(\arabic*),font=\normalfont,align=left,leftmargin=0pt,labelindent=0pt,
listparindent=\parindent,labelwidth=0pt,itemindent=!,topsep=3pt,parsep=0pt,itemsep=3pt,start=1}
\setlist[enumerate,2]{label=(\alph*),font=\normalfont,labelindent=*,leftmargin=*,topsep=3pt,start=1}
\setlist[itemize]{labelindent=*,leftmargin=*,topsep=5pt,itemsep=3pt}
\setlist[description]{labelindent=*,leftmargin=*,itemindent=-1 em}

\takeout{
\usepackage{seqsplit}
\usepackage{xstring}
\newcommand{\defaultshowkeysformat}[1]{%
\StrSubstitute{#1}{ }{\textvisiblespace}[\TEMP]%
\parbox[t]{\marginparwidth}{\raggedright\normalfont\small\ttfamily\(\{\){\color{red!50!black}\expandafter\seqsplit\expandafter{\TEMP}}\(\}\)}%
}

\renewcommand*\showkeyslabelformat[1]{%
\noexpandarg%
\defaultshowkeysformat{#1}%
}
}

\numberwithin{equation}{section}

\theoremstyle{plain}
\newtheorem{thm}{Theorem}[section]
\newtheorem{lem}[thm]{Lemma}
\newtheorem{propn}[thm]{Proposition}
\newtheorem{cor}[thm]{Corollary}

\theoremstyle{definition}
\newtheorem{defn}[thm]{Definition}
\newtheorem{expl}[thm]{Example}

\newtheorem{rem}[thm]{Remark}
\newtheorem{notn}[thm]{Notation}

\newcommand{\overbar}[1]{\mkern 1.5mu\overline{\mkern-2.5mu#1\mkern-1.5mu}\mkern 1.5mu}
\newcommand{\mybar}[3]{%
  \mathrlap{\hspace{#2}\overline{\scalebox{#1}[1]{\hphantom{\ensuremath{#3}}\vphantom{\rule{0pt}{5.5pt}}}}}\ensuremath{#3}
}

\DeclareMathOperator{\Ker}{\mathsf{ker}}

\newcommand{\infgameeq}{\sim_{\mathcal{G}}}
\newcommand{\barM}{\mybar{0.8}{2pt}{M}_1}

\newcommand{\barZ}{\mybar{0.7}{1.75pt}{Z}}

\DeclareMathOperator{\Fun}{\mathsf{Fun}}

\newcommand{\eps}{\varepsilon}
\newcommand{\syn}{\mathsf{syn}}

\newcommand{\id}{\mathsf{id}}
\newcommand{\subto}{\hookrightarrow}
\newcommand{\epito}{\twoheadrightarrow}
\newcommand{\monoto}{\rightarrowtail}
\newcommand{\inl}{\mathsf{inl}}
\newcommand{\inr}{\mathsf{inr}}
\newcommand{\inj}{\mathsf{in}}
\newcommand{\xra}[1]{\xrightarrow{~#1~}}
\newcommand{\xla}[1]{\xleftarrow{~#1~}}

\newcommand{\M}{\mathbb{M}}                   
\newcommand{\T}{\mathcal{T}}                     
\newcommand{\arity}{\mathsf{ar}}

\newcommand{\CatC}{\mathscr{C}}           
\newcommand{\CatD}{\mathscr{D}}
\newcommand{\CatK}{\mathscr{K}}
\newcommand{\CatI}{\mathscr{I}}
\newcommand{\CatG}{\mathscr{G}}
\newcommand{\K}{\mathscr{K}}

\newcommand{\Set}{\mathbf{Set}}

\newcommand{\Mon}{\mathscr{M}}             
\DeclareMathOperator{\Alg}{\mathsf{Alg}}

\setul{-1pt}{.4pt}
\DeclareMathOperator{\Support}{\mathsf{supp}}
\DeclareMathOperator{\Terms}{\mathscr{T}}
\newcommand{\Termsarg}[1]{\Terms_{\!#1}}
\DeclareMathOperator{\Win}{\mathsf{Win}}

\newcommand{\N}{\mathds{N}}            

\newcommand{\A}{\mathcal{A}}            

\newcommand{\CalG}{\mathcal{G}}             
\DeclareMathOperator{\colim}{\mathsf{colim}}    

\newcommand{\CalS}{\mathcal{S}}        
\newcommand{\bbD}{\mathbb{D}} 

\newcommand{\E}{\mathcal{E}}   

\DeclareMathOperator{\JSL}{\mathsf{JSL}}

\newcommand{\PT}{\mathsf{PT}}

\newcommand{\dl}{\mathbf 0}

\newcommand{\Id}{\mathsf{Id}}                 
\newcommand{\pow}{\mathscr P}              
\newcommand{\powf}{\pow_{\mathsf f}}
\newcommand{\CalN}{\mathcal{N}}           
\newcommand{\CalD}{\mathcal{D}}            
\newcommand{\Q}{\mathcal{Q}}                  

 \acmConference[]{}{}{}
 \acmYear{}
 \acmISBN{} 
 \acmDOI{} 
\startPage{1}

\setcopyright{none}

\usepackage{booktabs}   
\usepackage{subcaption} 

\begin{document}

\title[Graded Monads and Behavioural Equivalence Games]{Graded Monads and Behavioural Equivalence Games}         



\author[C. Ford]{Chase Ford}
\authornote{Supported by Deutsche Forschungsgemeinschaft (DFG, German Research Foundation) as
part of the Research and Training Group 2475 "Cybercrime and Forensic Computing"
(grant n. 393541319/GRK2475/1-2019).}          
\orcid{0000-0003-3892-5917}             
\affiliation{
  \institution{Friedrich-Alexander-Universit\"{a}t Erlangen-N\"{u}rnberg}            
  \country{Germany}                    
  }
\email{chase.ford@fau.de}          

\author[H. Beohar]{Harsh Beohar}
\orcid{0000-0001-5256-1334}             
\affiliation{
  \institution{University of Sheffield}            
  \country{United Kingdom}                    
}
\email{h.beohar@sheffield.ac.uk}          

\author[B. K\"{o}nig]{Barbara K\"{o}nig}
\authornote{Funded by the Deutsche Forschungsgemeinschaft (DFG, German
Research Foundation) -- project number 434050016.}    

\orcid{0000-0002-4193-2889}             
\affiliation{
  \institution{Universit\"{a}t Duisberg-Essen}            
  \country{Germany}                    
}
\email{barbara_koenig@uni-due.de}          

\author[S. Milius]{Stefan Milius}
\authornote{Funded by the Deutsche Forschungsgemeinschaft (DFG, German
  Research Foundation) -- project number 419850228.}          
\orcid{0000-0002-2021-1644}             
\affiliation{
  \institution{Friedrich-Alexander-Universit\"{a}t Erlangen-N\"{u}rnberg}            
  \country{Germany}                    
}
\email{stefan.milius@fau.de}          

\author[L. Schr\"{o}der]{Lutz Schr\"{o}der$^\dagger$}
  
\orcid{0000-0002-3146-5906}             
\affiliation{
  \institution{Friedrich-Alexander-Universit\"{a}t Erlangen-N\"{u}rnberg}            
  \country{Germany}                    
}
\email{lutz.schroeder@fau.de}          

\begin{abstract}
The framework of \emph{graded semantics} uses graded monads to capture behavioural equivalences of varying granularity, for example as found in the linear-time / branching-time spectrum, over general system types. We describe a generic Spoi\-ler-Duplicator game for graded semantics that is extracted from the given graded monad, and may be seen as playing out an equational proof; instances include standard pebble games for simulation and bisimulation as well as games for trace-like equivalences and coalgebraic behavioural equivalence. Considerations on an infinite variant of such games lead to a novel notion of \emph{infinite-depth graded semantics}. Under reasonable restrictions, the infinite-depth graded semantics associated to a given graded equivalence can be characterized in terms of a determinization construction for coalgebras under the equivalence at hand.
  \takeout{
  The framework of \emph{graded semantics} uses graded monads
  to capture behavioural equivalences with varying degrees of granularity
  over general system types,
  e.g.~as found on the linear-time / branching-time spectrum. We describe
  a generic Spoiler-Duplicator game for graded semantics that is extracted
  from the given graded monad, and may be seen as playing out an
  equational proof; instances include standard pebble games for
  (bi)simulation as well as games and for trace-like equivalences and
  coalgebraic behavioural equivalence. Furthermore, an infinite variant of
  such games leads to a novel notion of \emph{infinite-depth graded
    semantics}. Under reasonable restrictions, the infinite-depth
  semantics associated to a given graded equivalence can be
  characterized in terms of a determinization construction for
  coalgebras under the equivalence at hand.}
\end{abstract}



\maketitle

\section{Introduction}

\noindent The classical linear-time / branching-time
spectrum~\cite{Glabbeek90} organizes a plethora of notions of
behavioural equivalence on labelled transition systems at various
levels of granularity ranging from (strong) bisimilarity to trace
equivalence. Similar spectra appear in other system types, e.g.~on
probabilistic systems, again ranging from branching-time equivalence
such as probabilistic bisimilarity to linear-time ones such as
probabilistic trace equivalence~\cite{JouSmolka90}. While the
variation in system types (nondeterministic, probabilistic, etc.) is
captured within the framework of \emph{universal
  coalgebra}~\cite{Rutten00}, the variation in the granularity of
equivalence, which we shall generally refer to as the \emph{semantics}
of systems, has been tackled, in coalgebraic generality, in a variety
of approaches~\cite{HJS07,KR15,JSS15,JLR18}. One setting that
manages to accommodate large portions of the
linear-time / branching-time spectrum, notably including also
intermediate equivalences such as ready similarity, is based on
\emph{graded monads}~\cite{MPS15,DMS19,FMS21a}.

An important role in the theoretical and algorithmic treatment of a
behavioural equivalence is classically played by equivalence
games~\cite{Glabbeek90,Stirling99}, e.g.~in partial-order
techniques~\cite{hnw:por-bisimulation-checking} or in on-the-fly
equivalence checking~\cite{h:bisim-verif-journal}. In the present
work, we contribute to \emph{graded semantics} in the sense indicated
above by showing that, under mild conditions, we can extract from a
given graded monad a Spoiler-Duplicator game~\cite{Stirling99} that
characterizes the respective equivalence, i.e.~ensures that two states
are equivalent under the semantics iff Duplicator wins the game.

As the name suggests, graded monads provide an \emph{algebraic} view
on system equivalence; they correspond to \emph{grad\-ed theories},
i.e.~algebraic theories equipped with a notion of \emph{depth} on
their operations. It has been noticed early on~\cite{MPS15} that many
desirable properties of a semantics depend on this theory being
\emph{depth-1}, i.e.~having only equations between terms that are
uniformly of depth~1. Standard examples include distribution of
actions over non-deterministic choice (trace semantics) or
monotonicity of actions w.r.t.~the choice ordering
(similarity)~\cite{DMS19}. Put simply, our generic equivalence game
plays out an equational proof in a depth-1 equational theory in a
somewhat nontraditional manner: 
Duplicator starts a round by playing a set of equational assumptions
she claims to hold at the level of successors of the present state,
and Spoiler then challenges one of these assumptions.
              										
In many concrete cases, the game can be rearranged in a straightforward
manner to let Spoiler move first as usual; in this view, the equational claims
of Duplicator roughly correspond to a short-term strategy determining the
responses she commits to playing after Spoiler's next move. In particular,
the game instantiates, after such rearrangement, to the standard pebble
game for bisimilarity. We analyse additional cases, including similarity
and trace equivalence, in more detail. In the latter case, several natural
variants of the game arise by suitably restricting strategies played by
Duplicator.

It turns out that the game is morally played on a form of
pre-determinization of the given coalgebra, which lives in the
Eilenberg-Moore category of the zero-th level of the graded monad, and as
such generalizes a determinization construction that applies in
certain instances of coalgebraic language semantics of
automata~\cite{JSS15}.  Under suitable conditions on the graded monad,
this pre-determinization indeed functions as an actual
determinization, i.e.~it turns the graded semantics into standard
coalgebraic behavioural equivalence for a functor that we construct on
the Eilenberg-Moore category. This construction simultaneously
generalizes, for instance, the standard determinization of serial labelled
transition systems for trace equivalence and the identification of
similarity as behavioural equivalence for a suitable functor on
posets~\cite{KKV12} (specialized to join
semilattices).

While graded semantics has so far been constrained to apply only to
finite-depth equivalences (finite-depth bisimilarity, finite trace
equivalence, etc.), we obtain, under the mentioned conditions on the
graded monad, a new notion of infinite-depth equivalence induced by a
graded semantics, namely via the (pre-)determinization. It turns out
the natural infinite version of our equivalence game captures
precisely this infinite-depth equivalence. This entails 
a fixpoint characterization of graded semantics on finite systems,
giving rise to perspectives for a generic algorithmic treatment.

\paragraph{Related Work.} Game characterizations of process
equivalences are an established theme in concurrency theory; 
they tend to be systematic but not generic~\cite{Glabbeek90,CD08}.
Work on games for spectra of quantitative equivalences is positioned
similarly~\cite{FLT11,fl:quantitative-spectrum-journal}. The idea of developing \mbox{(bi)}simulation
games in coalgebraic generality goes back to work on branching-time
simulations based on relators~\cite{Baltag00}. There
is recent highly general work, conducted in a fibrational setting, on
so-called codensity games for various notions of
bisimilarity~\cite{kkhkh:codensity-games}. The emphasis in this work
is on generality w.r.t.~the measure of bisimilarity, covering,
e.g.~two-valued equivalences, metrics, pre-orders, and topologies,
while, 
viewed through the lens of spectra of equivalences, the
focus remains on branching time. The style of the codensity game is
inspired by modal logic, in the spirit of coalgebraic Kantorovich
liftings~\cite{BaldanEA18,WildSchroder20}; 
Spoiler plays predicates thought of as arguments of modalities.
Work focused more specifically on games for Kantorovich-style
coalgebraic behavioural equivalence and behavioural
metrics~\cite{km:bisim-games-logics-metric} similarly concentrates on
the branching-time case. A related game-theoretic characterization is
implicit in work on $\Lambda$-(bi)similarity~\cite{GorinSchroder13},
also effectively limited to branching-time. Comonadic game
semantics~\cite{ADW17, AS18, CD21} proceeds in the opposite way
compared to the mentioned work and ours: It takes existing games as
the point of departure, and then aims to develop categorical models.

Graded semantics was developed in a line of work mentioned
above~\cite{MPS15,DMS19,FMS21a}. The underlying notion of
graded monad stems from algebro-geometric work~\cite{Smirnov08}
and was introduced into computer science (in substantially higher generality)
in work on the semantics of effects~\cite{Katsumata14}. Our pre-determinization
construction relates to work on coalgebras over algebras~\cite{BK11}.

%
%

\paragraph*{Organization.} We discuss preliminaries on categories,
coalgebras, graded monads, and games in \cref{sec:prelims}. We
recall the key notions of graded algebra and canonical graded algebra
in~\cref{sec:algebras}, and graded semantics in~\cref{sec:semantics}.
We introduce our pre-determinization construction in \cref{sec:determinization},
and finite behavioural equivalence games in \cref{S:games}. In
\cref{sec:infinte-depth}, we consider the infinite version of the
game, relating it to behavioural equivalence on the pre-determinization. We
finally consider specific cases in detail in~\cref{sec:cases}.

\section{Preliminaries}\label{sec:prelims}
We assume basic familiarity with category theory~\cite{AHS90}. We
will review the necessary background on coalgebra~\cite{Rutten00},
graded monads~\cite{Smirnov08,MPS15}, and the standard bisimilarity
game~\cite{Stirling99}.

\paragraph*{The category of sets.}
Unless explicitly mentioned otherwise, we will
work in the category $\Set$ of sets and functions
(or \emph{maps}), which is both complete and
cocomplete. We fix a terminal object $1=\{\star\}$
and use $!_{X}$ (or just $!$ if confusion is
unlikely) for the unique map $X\to 1$.

In the subsequent sections, we will mostly draw examples from (slight
modifications of) the following (endo-)functors on $\Set$. The
\emph{powerset functor} $\pow$ sends each set $X$ to its set of
subsets $\pow X$, and acts on a map $f\colon X\to Y$ by taking direct
images, i.e.~$\pow f(S):= f[S]$ for $S\in\pow X$.
We write $\powf$ for the \emph{finitary powerset functor} which sends
each set to its set of finite subsets; the action of $\powf$ on maps is
again given by taking direct images. Similarly, $\pow^+$ denotes the
non-empty powerset functor
($\pow^+(X)=\{Y\in\pow(X)\mid Y\neq\emptyset\}$), and $\powf^+$ its
finitary subfunctor
($\powf^+(X)=\{Y\in\powf(X)\mid Y\neq\emptyset\}$).

We write $\CalD X$ for the set of \emph{distributions} on a set $X$:
maps $\mu\colon X\to [0,1]$ such that $\sum_{x\in X}\mu(x)=1$.  A
distribution $\mu$ is \emph{finitely supported} if the set
$\{x\in X~|~\mu(x)\neq 0\}$ is finite. The set of finitely supported
distributions on $X$ is denoted $\CalD_f X$. The assignment
$X\mapsto\CalD X$ is the object-part of a functor: given
$f\colon X\to Y$, the map $\CalD f\colon\CalD X\to \CalD Y$ assigns to a
distribution $\mu\in\CalD X$ the \emph{image} distribution
$\CalD f(\mu)\colon Y\to [0, 1]$ defined by
$\CalD f(\mu)(y)=\sum_{x\in X\mid f(x)=y} \mu(x)$.  Then,
$\CalD f(\mu)$ is finitely supported if $\mu$ is, so $\CalD_f$ is
functorial as well.

\takeout{
We will also consider the \emph{contravariant powerset functor}
$\Q\colon\Set\to\Set^{op}$, which acts on sets according to $\pow$,
i.e.~$\Q(X):= \pow X$, but sends a map $f\colon X\to Y$ to the
\emph{inverse image map} $\Q f\colon\pow Y\to \pow X$ defined by
$S\mapsto f^{-1}[S]$. The composite functor
$\CalN:=(\Set\xra{\Q}\Set^{op}\xra{\Q^{op}}\Set)$ is called the
\emph{neighborhood functor}.  Explicitly, $\CalN$ is the endofunctor
on $\Set$ which sends a set $X$ to the set $\pow(\pow X)$ of all
\emph{neighborhoods} of~$X$, and it sends a map $f\colon X\to Y$
to the map $\CalN f\colon\CalN X\to \CalN Y$ defined by the
assignment
$\mathscr S\mapsto\{S\in \pow Y~|~f^{-1}[S]\in \mathscr S\}$.
}
\takeout{
We will also work with the \emph{(finitely supported)
distribution functor} $\CalD\colon\Set\to\Set$ is the functor
which sends a set $X$ to the set of \emph{(finitely supported)
probability distributions on $X$}. That is, $\CalD X$ consists of
all maps $\varphi\colon X\to [0,1]$ (where $[0,1]$ denotes the
real unit interval) such that
\[
\sum_{x\in X}\varphi(x) = 1
\]
and such that the \emph{support}
$\Support(\varphi):=\{x\in X\mid\varphi(x)\ne 0\}$ is finite. It
is well known (and easy to see) that each $\varphi\in\CalD X$
is equivalently presented as the formal convex sum
\[
\sum_{x\in\Support(\varphi)} \varphi(x)\cdot x.
\]
In this notation, the action $\CalD f\colon \CalD X\to\CalD Y$
on a map $f\colon X\to Y$ is conveniently described via the
assignment
$\sum_{i\leq k}r_i\cdot x_i \mapsto \sum_{i\leq k} r_i\cdot f(x_i).$
}
%
\noindent
\paragraph*{Coalgebra.}
 We will review the
basic definitions and results of \emph{universal
coalgebra}~\cite{Rutten00}, a categorical framework for the uniform
treatment of a variety of reactive system types.
\begin{defn}
  For an endofunctor $G\colon\CatC\to\CatC$ on a
  category $\CatC$, a \emph{$G$-coalgebra} (or just
  \emph{coalgebra}) is a pair $(X, \gamma)$ consisting
  of an object $X$ in $\CatC$ and a morphism
  $\gamma\colon X\to GX$.
  A \emph{(coalgebra) morphism} from $(X, \gamma)$ to
  a coalgebra $(Y, \delta)$ is a morphism $h\colon X\to Y$
  such that $\delta\cdot h = Fh\cdot\gamma$.
\end{defn}
\noindent Thus, for $\CatC = \Set$, a coalgebra consists of a set $X$
of \emph{states} and a map $\gamma\colon X\to GX$, which we view as a
transition structure that assigns to each state $x\in X$ a structured
collection $\gamma(x)\in GX$ of \emph{successors} in $X$.
\begin{expl}\label{E:coalg}
  We describe some examples of functors on $\Set$ and their coalgebras
  for consideration in the subsequent. Fix a finite set $\A$ of
  \emph{actions}.
\begin{enumerate}
\item\label{E:coalg:1} Coalgebras for the functor $G=\pow(\A\times -)$
  are just 
  \emph{$\A$-labelled transition systems (LTS)}: Given such a
  coalgebra $(X, \gamma)$, we can view the elements
  $(a, y)\in\gamma(x)$ as the $a$-successors of~$x$. We call
  $(X, \gamma)$ \emph{finitely branching} (resp.~\emph{serial}
  if~$\gamma(x)$ is finite (resp.~non-empty) for all~$x\in X$.
 Finitely branching (resp.~serial) LTS are coalgebras for
 the functor $G=\powf(\A\times -)$ (resp. $\pow^+(\A\times -)$).

\takeout{
\item A coalgebra for the neighborhood functor $\CalN$ is a
  \emph{neighborhood frame}~\cite{HKP09}: a map
  $\nu\colon X\to \CalN X$ assigning each state $x\in X$ to its
  set of \emph{neighbourhoods} $\nu(x)\in\CalN X$.%
  \smnote{I wouldn't write $\pow(\pow X)$; this gives the wrong
    impression that $\N$ is the double covariant powerset functor.}
  A morphism
  of neighborhood frames (also: \emph{bounded morphism}) from
  $(X, \nu)$ to $(X', \nu')$ is a map $f\colon X\to X'$ such that for
  all $x\in X$: $f^{-1}[V]\in\nu(x)$ iff $V\in\nu'(f(x))$ for all
  $V\in\pow X'$.
}

\item A coalgebra $(X,\gamma)$ for the functor $G=\CalD(\A\times -)$
  is a \emph{(generative) probabilistic transition system} (PTS): The
  transition structure~$\gamma$ assigns to each state $x\in X$ a
  distribution $\gamma(x)$ on pairs $(a,y)\in\A\times X$. We think of
  $\gamma(x)(a, y)$ as the probability of executing an $a$-transition
  to state~$y$ while sitting in state~$x$.
  A PTS $(X, \gamma)$ is \emph{finitely branching} if $\gamma(x)$ is
  finitely supported for all $x\in X$; then, finitely branching PTS
  are coalgebras for $\CalD_f(\A\times -)$.
\end{enumerate}
\end{expl}

Given coalgebras $(X, \gamma)$ and $(Y, \delta)$ for an endofunctor
$G$ on $\Set$, states $x\in X$ and $y\in Y$ are
\emph{$G$-behaviourally equivalent} if there exist coalgebra morphisms
\[
  (X,\gamma) \xra{f} (Z,\zeta) \xla{g} (Y,\delta)
\]
such that $f(x)= g(y)$. Behavioural equivalence can be approximated
via the (initial $\omega$-segment of the) \emph{final chain}
$(G^n1)_{n\in\omega}$, where $G^n$ denotes $n$-fold application of
$G$. 
The \emph{canonical cone} of a coalgebra
$(X, \gamma)$ is then the family of maps $\gamma_n\colon X\to G^n1$
defined inductively for $n\in\omega$ by
\begin{align*}
  \gamma_0 &= \big(X\xra{!} 1\big), \text{and} \\
  \gamma_{n+1} &= \big(X \xra{\gamma} GX \xra{G\gamma_n} GG^n 1 = G^{n+1}1\big).
\end{align*}
States $x, y\in X$ are \emph{finite-depth behaviourally equivalent} if
$\gamma_n(x)= \gamma_n(y)$ for all $n \in \omega$.

\begin{rem}\label{rem:finite-depth}
  It follows from results of Worrell~\cite{Worrell05} that behavioural
  equivalence and finite-depth behavioural equivalence coincide for
  finitary functors on $\Set$, where a functor $G$ on $\Set$ is
  \emph{finitary} if it preserves filtered colimits. Equivalently, for
  every set $X$ and each $x \in GX$ there exists a finite subset $Y
  \subseteq X$ such that $x = Gi[GY]$, where $i\colon Y \subto X$ is
  the inclusion map~\cite[Cor.~3.3]{amsw19-1}.
\end{rem}

\paragraph*{Bisimilarity games.}
We briefly recapitulate the classical \emph{bisimilarity game}, a
two-player graph game between the players Duplicator (D) and Spoiler
(S); player~D tries to show that two given states are bisimilar,
while~S tries to refute this. \emph{Configurations} of the game are
pairs $(x, y)\in X\times X$ of states in a LTS $(X, \gamma)$. The game
proceeds in rounds, starting from the \emph{initial configuration},
which is just the contested pair of states. In each round, starting
from a configuration $(x, y)$,~S picks one of the sides, say,~$x$, and
then selects an action $a\in\A$ and an $a$-successor~$x'$ of~$x$;
player~D then selects a corresponding successor on the other side, in
this case an $a$-successor~$y'$ of~$y$. The game then reaches the new
configuration $(x', y')$. If a player gets stuck, the play is
\emph{winning} for their opponent, whereas any infinite play is
winning for~D.

It is well known (e.g.~\cite{Stirling99}) that~D has a
winning strategy in the bisimilarity game at a configuration $(x, y)$
iff $(x, y)$ is a pair of bisimilar states. Moreover, for finitely
branching LTS, an equivalent formulation may be given in terms of the
\emph{$n$-round bisimilarity game}:
the rules of the $n$-round game are the same as those above, only
now~D wins as soon as at most~$n$ rounds have been played. In fact, a
configuration $(x, y)$ is a bisimilar pair precisely if~D has
a winning strategy in the $n$-round bisimilarity game for all
$n\in\omega$.

We mention just one obvious variation of this game that characterizes
a different spot on the linear-time/branching-time spectrum: The
\emph{mutual-simulation game} is set up just like the bisimulation
game, except that~S may only choose his side once, in the first round,
and then has to move on that side in all subsequent rounds (in the
bisimulation game, he can switch sides in every round if he
desires). It is easily checked that states~$x,y$ are mutually similar
iff~S wins the position $(x,y)$ in the mutual-simulation game. We will
see that both these games (and many others) are obtained
as instances of our generic notion of graded equivalence game.

\paragraph*{Graded monads.}
We now review some background material on graded monads
\cite{Smirnov08, MPS15}:
%
\begin{defn}\label{D:gradedmonad}
  A \emph{graded monad} $\M$ on a category $\CatC$ is a triple
  $(M, \eta, \mu)$ where $M$ is a family of functors
  $M_n\colon\CatC\to\CatC$ on $\CatC$ ($n\in\omega$),
  $\eta\colon \id \to M_0$ is a natural transformation (the
  \emph{unit}), and $\mu$ is a family of natural transformations%
  \begin{equation}
    \mu^{n,k}\colon M_nM_k\to M_{n+k}  \tag{$n,k\in\omega$}
  \end{equation}
  (the \emph{multiplication}) such that the following diagrams commute
  for all $n,m,k\in\omega$:
  \begin{equation}\label{diagram:unitlaw}
    \begin{tikzcd}[column sep=40]
      &
      M_n
      \arrow[ld, "M_n\eta"'] \arrow[rd, "\eta M_n"]  \arrow[d, "\Id"]
      \\
      M_nM_0
      \arrow[r, "{\mu^{n,0}}"]
      &
      M_n
      &
      M_0M_n \arrow[l, "{\mu^{0,n}}"']
    \end{tikzcd}
  \end{equation}
  \begin{equation}\label{diagram:associativelaw}
    \begin{tikzcd}[column sep = 60]
      M_nM_kM_m
      \arrow[r, "{M_n\mu^{k,m}}"]
      \arrow[d, "{\mu^{n,k}M_m}"']
      &
      M_{n}M_{k+m} \arrow[d, "{\mu^{n,k+m}}"]
      \\
      M_{n+k}M_m \arrow[r, "{\mu^{n+k,m}}"]
      &
      M_{n+k+m}
    \end{tikzcd}
  \end{equation}
  We refer to~\eqref{diagram:unitlaw} and~\eqref{diagram:associativelaw}
  as the \emph{unit} and \emph{associative} laws of $\M$,
  respectively. We call $\M$ \emph{finitary} if all of the
  functors $M_n\colon\CatC\to\CatC$ are finitary.
\end{defn}
\noindent
The above notion of graded monad is due to Smirnov \cite{Smirnov08}.
Katsumata~\cite{Katsumata14}, Fujii et al.~\cite{FKM16}, and
Mellies~\cite{Mellies17} consider a more general notion of graded (or
\emph{parametrized}) monad given as a lax monoidal action of a
monoidal category $\Mon$ (representing the system of grades) on a
category $\CatC$. Graded monads in the above sense are recovered by
taking~$\Mon$ to be the (discrete category induced by the) monoid
$(\N, +, 0)$.

The graded monad laws imply that the triple $(M_0, \eta, \mu^{0,0})$
is a (plain) monad on the base category $\CatC$; we use this freely
without further mention.

\begin{expl}\label{E:graded-monad}
  We review some salient constructions~\cite{MPS15} of graded
  monads on $\Set$ for later use.
  \begin{enumerate}
  \item\label{E:graded-monad:1} Every endofunctor $G$ on $\Set$
    induces a graded monad~$\M_G$ with underlying endofunctors
    $M_n:= G^n$ (the $n$-fold composite of~$G$ with itself); the unit
    $\eta_X\colon X\to G^0X= X$ and multiplication
    $\mu_X^{n,k}\colon G^nG^kX\to G^{n+k}X$ are all identity maps.
    We will later see that~$\M_G$ captures (finite-depth) $G$-behavioural
    equivalence.

  \item\label{item:graded-kleisli}
  Let $(T, \eta, \mu)$ be a monad on $\Set$, let $F$ be
  an endofunctor on $\Set$, and let $\lambda\colon FT\to TF$
  be a natural transformation such that
    \[
      \lambda\cdot F\eta=\eta F
      \qquad\text{and}\qquad
      \lambda \cdot F\mu = \mu F \cdot T\lambda \cdot \lambda T
    \]
    (i.e.~$\lambda$ is a distributive law of the functor~$F$ over the
    monad~$T$). For each $n\in\omega$, let
    $\lambda^n\colon F^nT\to TF^n$ denote the natural transformation
    defined inductively by
    \[
      \lambda^0:= \id_T; \qquad \lambda^{n+1}:= \lambda^n F\cdot F^n\lambda.
    \]
    We obtain a graded monad with $M_n:= TF^n$, unit $\eta$, and
    components $\mu^{n,k}$ of the multiplication given as the
    composites
    \[
      TF^nTF^k\xra{T\lambda^n F^k}  TTF^nF^k = TTF^{n+k}\xra{ \mu F^{n+k}} TF^{n+k}.
    \]
    Such graded monads relate strongly to Kleisli-style coalgebraic
    trace semantics~\cite{HJS07}.
  \item\label{item:T-traces} We obtain (by instance of the example above) a graded monad
    $\M_T(\A)$ with $M_n= T(\A^n\times -)$ for every monad $T$ on
    $\Set$ and every set~$\A$. Thus, $\M_T$ is a graded monad for
    traces under effects specified by~$T$; e.g.~for $T=\CalD$, we will
    see that $\M_T(\A)$ captures probabilistic trace equivalence on
    PTS.

  \item\label{E:graded-monad:4}Similarly, given a monad $T$, an
    endofunctor $F$, both on the same category $\CatC$, and a
    distributive law $\lambda\colon TF \to FT$ of~$T$ over $F$, we
    obtain a graded monad with $M_n := F^nT$, unit and
    multiplication given analogously as in
    item~\ref{item:graded-kleisli} above
    (see~\cite[Ex.~5.2.6]{MPS15}). Such graded monads relate strongly
    to Eilenberg-Moore-style coalgebraic language semantics~\cite{bms13}.
\end{enumerate}
\end{expl}

Graded variants of Kleisli triples have been introduced and proved
equivalent to graded monads (in a more general setting)
by Katsumata~\cite{Katsumata14}:

\begin{notn}\label{N:star}
We will employ the \emph{graded Kleisli star} notation:
for $n\in\omega$ and a morphism $f\colon X\to M_k Y$, we
write%
\begin{equation}\label{Eqn:Kleisli-star}
  f^*_n:=  \big(M_nX\xra{M_nf} M_nM_k\xra{\mu^{n,k}}M_{n+k}Y\big).
\end{equation}
In this way, we obtain a morphism satisfying the following graded
variants~\cite[Def.~2.3]{Katsumata14} of the usual laws of the
Kleisli star operation for ordinary monads: for every $m\in\omega$ and
morphisms $f\colon X\to M_nY$ and $g\colon Y\to M_kZ$ we have:%
\begin{align}
  f^*_0\cdot\eta_X &= f,	\label{item:star-2} \\
  (\eta_X)^*_n &= \id_{M_n X}, \label{item:start-3}\\
  (g^*_{n}\cdot f)^*_m &= g^*_{m+n}\cdot f_m^*.    \label{item:star-1}
\end{align}

\end{notn}

\paragraph*{Graded theories.}
Graded theories, in a generalized form in which arities of operations
are not restricted to be finite, have been proved equivalent to graded
monads on $\Set$~\cite{MPS15} (the finitary case was implicitly
covered already by Smirnov~\cite{Smirnov08}).  We work primarily with
the finitary theories below; we consider infinitary variants of such
theories only when considering infinite-depth equivalences
(\cref{sec:infinte-depth}).

\begin{defn}\label{def:theory}
\begin{enumerate}
\item A \emph{graded signature} is a set $\Sigma$ of \emph{operations}
  $f$ equipped with a finite \emph{arity} $\arity(f) \in \omega$ and a
  finite \emph{depth} $d(f)\in\omega.$ An operation of arity 0 is
  called a \emph{constant}.

\item
Let $X$ be a set of \emph{variables} and let
$n\in\omega$. The set $\Termsarg{\Sigma, n}(X)$
of \emph{$\Sigma$-terms of uniform depth $n$ with
variables in $X$} is defined inductively as follows:
every variable $x\in X$ is a term of uniform depth
$0$ and, for $f\in\Sigma$ and
$t_1,\dots, t_{\arity(f)}\in\Termsarg{\Sigma, k}(X)$,
$f(t_1,\dots, t_{\arity(f)})$ is a $\Sigma$-term
of uniform depth $k+ d(f)$. In particular,
a constant $c$ has uniform depth $k$ for all $k\geq d(c)$.
\item A \emph{graded $\Sigma$-theory} is a set $\E$ of
  \emph{uniform-depth equations}: pairs $(s, t)$, written `$s=t$',
  such that $s,t\in\Termsarg{\Sigma, n}(X)$ for some $n\in\omega$;
  %
  we say that $(s, t)$ is \emph{depth-$n$}.
  A theory is \emph{depth-$n$} if all of its equations and
  operations have depth at most $n$.
\end{enumerate}
\end{defn}

\begin{notn}\label{N:substitution}
A \emph{uniform-depth substitution} is a map
$\sigma\colon X\to\Termsarg{\Sigma, k}(Y)$,
where $k\in\omega$ and $X, Y$ are sets.
Then $\sigma$ extends to a family of maps
$\bar{\sigma}_n\colon\Termsarg{\Sigma, n}(X)\to\Termsarg{\Sigma, k+n}(Y)$
($n\in\omega$) defined recursively by
\[
\bar{\sigma}_n(f(t_1,\dots, t_{\arity(f)})) =
f(\bar{\sigma}_m(t_1),\dots, \bar{\sigma}_m(t_{\arity(f)})),
\]
where $t_i\in\Termsarg{\Sigma, m}$ and
$d(f)+m = n$. For a term $t\in T_{\Sigma, k}(X)$,
we also write $t\sigma:= \bar{\sigma}_n(t)$ when
confusion is unlikely.
\end{notn}
\noindent
Given a graded theory $\T=(\Sigma, \E)$, we have essentially the
standard notion of equational derivation (sound and complete over
graded algebras, cf.\ \cref{sec:algebras}), restricted to
uniform-depth equations. Specifically, the system includes the
expected rules for reflexivity, symmetry, transitivity, and
congruence, and moreover allows substituted introduction of axioms: If
$s=t$ is in~$\E$ and~$\sigma$ is a uniform-depth substitution, then
derive the (uniform-depth) equation $s\sigma=t\sigma$. (A substitution
rule that more generally allows uniform-depth substitution into
derived equations is then admissible.)  For a set~$Z$ of uniform-depth
equations, we write
\begin{equation*}
Z\vdash s=t
\end{equation*}
if the uniform-depth equation $s=t$ is derivable from equations in~$Z$
in this system; note that unlike the equational axioms in~$\E$, the
equations in~$Z$ cannot be substituted into in such a derivation (they
constitute assumptions on the variables occurring in~$s,t$).

We then see that~$\T$ induces a graded monad~$\M_{\T}$ with $M_nX$
being the quotient of $\Termsarg{\Sigma, n}(X)$ modulo derivable
equality under $\E$; the unit and multiplication of $\M_{\T}$ are
given by the inclusion of variables as depth-0 terms and the
collapsing of layered terms, respectively. Conversely, every graded
monad arises from a graded theory in this way~\cite{MPS15}.

We will restrict attention to graded monads presented by depth-$1$
graded theories:

\begin{defn}
  A \emph{presentation} of a graded monad $\M$ is a graded theory $\T$
  such that $\M\cong\M_{\T}$, in the above notation.  A graded monad
  is \emph{depth-1} if it has a depth-$1$ presentation.
\end{defn}

\begin{expl}\label{E:graded-theory}
  Fix a set $\A$ of actions. We describe \mbox{depth-$1$} graded theories
  associated (via the induced behavioural equivalence,
  \cref{sec:semantics}) to standard process equivalences on LTS and
  PTS~\cite{DMS19}.
\begin{enumerate}
\item\label{item:jsl-a} The graded theory $\JSL(\A)$ of
  \emph{$\A$-labelled join semilattices} has as depth-1 operations all
  formal sums
  \[
    \textstyle\sum_{i=1}^na_i(-),
    \quad
    \text{for \mbox{$n\ge 0$} and $a_1,\dots,a_n\in\A$}
  \]
  (and no depth-$0$ operations); we write~$0$ for
  the empty formal sum. The axioms of $\JSL(\A)$ consist of all
  depth-1 equations $\sum_{i=1}^na_i(x_i) = \sum_{j=1}^m b_j(y_j)$
  (where the~$x_i$ and~$y_j$ are variables, not necessarily distinct)
  such that
  $\{(a_i, x_i)~|~1\le i\leq n\}=\{(b_j, y_j)~|~1\le j\leq m\}$.  The
  graded monad induced by $\JSL(\A)$ is $\M_G$ for
  $G=\powf(\A\times(-))$
  (cf.~\cref{E:graded-monad}.\ref{E:graded-monad:1}).

\item\label{item:pt-a} The \emph{graded theory of probabilistic
    traces}, $\PT(\A)$, has a depth-0 convex sum operation
  \[
    \textstyle\sum^n_{i = 1} p_i\cdot(-)
    \quad
    \text{for all $p_1, \ldots, p_n \in [0,1]$ such that
      $\sum^n_{i=1} p_i = 1$}
  \]
  and unary depth-1 operations $a(-)$ for all actions $a\in\A$. As
  depth-0 equations, we take the usual equational axiomatisation of
  convex algebras, which is given by the equation
  $\sum^n_{i = 1}\delta_{ij}\cdot x_j = x_i$ (where $\delta_{ij}$
  denotes the Kronecker delta function) and all instances of the
  equation scheme
\[
\sum^n_{i = n} p_i\cdot \sum^m_{j=1} q_{ij}\cdot x_j =
\sum_{j=1}^m \Big(\sum^{n}_{i=1}p_iq_{ij}\Big)\cdot x_j.
\]
We further impose depth-1 equations
stating that actions distribute over
convex sums:
\[
a\Big(\sum_{i=1}^np_i\cdot x_i\Big) = \sum_{i=1}^n p_i\cdot a(x_i).
\]
The theory $\PT(\A)$ presents~$\M_{\CalD_f}(\A)$,
where $\CalD_f$ is the finitely supported distribution monad
(cf.~\cref{E:graded-monad}.\ref{item:T-traces}).

\item\label{item:traces-a} We mention two variations on the graded
  theory above. First, the \emph{graded theory of (non-deterministic)
    traces} presenting $\M_{\powf}(\A)$ has depth-0
  operations~$+$,~$0$ and equations for join-semilattices with bottom,
  and unary depth-1 operations~$a$ for $a\in\A$ as in
  \ref{item:jsl-a} above; the depth-1 equations now state that
  actions distribute over joins and preserve bottom. Second, the
  \emph{graded theory of serial (non-deterministic) traces} arises by
  omitting~$0$ and associated axioms from the graded theory of traces,
  and yields a presentation of $\M_{\powf^+}(\A)$.

\item\label{item:simulation-a} The \emph{graded theory of simulation}
  has the same signature and depth-0 equations as the graded theory of
  traces, 
  along with depth-1 equations stating that actions are monotone:
\[
a(x+y) + a(x) = a(x+y).
\]
The theory of simulation equivalence then
yields a presentation of the graded monad
with $M_nX$ defined inductively along with
a partial ordering as follows: We take
\mbox{$M_0X = \powf(X)$} ordered by set inclusion.
We equip $\A\times M_nX$ with the product ordering
of the discrete order on $\A$ and the given
ordering on $M_nX$. Then $M_{n+1}X = \powf^{\downarrow}(\A\times M_nX)$ is
the set  of downwards-closed finite subsets of
$\A\times M_nX$.
\end{enumerate}
\end{expl}

In the following lemma, an \emph{epi-transformation} is a natural
transformation $\alpha$ whose components $\alpha_X$ are surjective maps.

\begin{notheorembrackets}
\begin{lem}[{\cite{MPS15}}]\label{L:depth-1}
  A graded monad $\M$ on $\Set$ is \emph{depth-1} if and only if
  all $\mu^{1,n}$ are epi-transformations and the following is
  object-wise a coequalizer diagram in the category of
  Eilenberg-Moore algebras for the monad $M_0$ for all
  $n\in\omega$:
  \begin{equation}
    \begin{tikzcd}[column sep = 35]\label{Diagram:depth1}
      M_1M_0M_n \arrow[r, "M_1\mu^{0,n}", shift left] \arrow[r,
      "\mu^{1,0}M_n"', shift right] & M_1M_n \arrow[r, "\mu^{1,n}"] &
      M_{1+n}.
    \end{tikzcd}
  \end{equation}
\end{lem}
\end{notheorembrackets}

\section{Graded Behavioural Equivalences}\label{sec:semantics}
We next recall the notion of a graded semantics~\cite{MPS15}
on coalgebras for an endofunctor on $\Set$;
we illustrate several instantiations of subsequent interest.


\begin{defn}[Graded semantics]
  A \emph{(depth-1) graded semantics} for an endofunctor
  $G\colon\Set\to\Set$ is a pair $(\alpha, \M)$ consisting of a
  (depth-1) graded monad $\M$ on $\Set$ and a natural transformation
  $\alpha\colon G\to M_1$ .
\end{defn}
\noindent
Given a $G$-coalgebra $(X,\gamma)$, the graded
semantics $(\alpha, \M)$ induces a sequence of
maps $\gamma^{(n)}\colon X\to M_n1$ inductively
defined by
\begin{align*}
  \gamma^{(0)} &:= (X\xra{\eta_X}M_0X\xra{M_0!}M_01); \\
  \gamma^{(n+1)} &:= (X\xra{\alpha_X\cdot\gamma} M_1X
                                  \xra{M_1\gamma^{(n)}} M_1M_n 1
				 \xra{\mu^{1,n}_1} M_{1+n}1)
\end{align*}
(or, using the graded Kleisli star,
$\gamma^{(n+1)} = (\gamma^{(n)})^*_1\cdot\alpha_X\cdot\gamma$). We call
$\gamma^{(n)}(x)\in M_n1$ the \emph{$n$-step $(\alpha, \M)$-behaviour}
of $x\in X$.

\begin{defn}[Graded behavioural equivalence]
  States ${x\in X, y\in Y}$ in $G$-coalgebras $(X, \gamma)$ and
  $(Y, \delta)$ are \emph{depth-n behaviourally equivalent} under
  $(\alpha, \M)$ if ${\gamma^{(n)}(x) = \delta^{(n)}(y)}$, and
  \emph{$(\alpha, \M)$-behaviourally equivalent} if
  $\gamma^{(n)}(x) = \delta^{(n)}(y)$ for all $n\in\omega$. We
  refer to $(\alpha,\M)$-behavioural equivalence as a
  \emph{graded behavioural equivalence} or just a \emph{graded
    equivalence}.
\end{defn}

\begin{expl}\label{E:semantics}
We recall~~\cite[Section 4]{DMS19} several graded
equivalences, restricting primarily to LTS and PTS.
\begin{enumerate}
\item\label{item:sem-beh}
 For an endofunctor $G$ on $\Set$, finite-depth
 $G$-behavioural equivalence arises as the
 graded equivalence with $\M = \M_G$ and
 $\alpha = \id$, where $\M_G$ is the graded
 monad of
 \cref{E:graded-monad}.\ref{E:graded-monad:1}.
 By~\cref{rem:finite-depth}, it follows that
 $(\id, \M_G)$ captures full coalgebraic bisimilarity
 in case $G$ is finitary.

\item\label{item:sem-trace} Let $(X, \gamma)$ be an LTS, let
  ${x\in X}$, and let ${w\in\A^*}$ be a finite word over $\A$. We
  write $x\xra{w} y$ if the state $y$ can be reached on a path whose
  labels form the word $w$. A \emph{finite trace} at $x\in X$ is a
  word $w\in\A^*$ such that $x\xra{w} y$ for some $y\in X$; the set of
  finite traces at $x$ is denoted $\tau(x)$. States $x,y\in X$ are
  trace equivalent if $\tau(x)=\tau(y)$. Trace equivalence on finitely
  branching LTS is captured by the graded equivalence induced by
  $\M=\M_{\powf}(\A)$
  (cf.~\cref{E:graded-monad}.\ref{item:T-traces}), again with
  $\alpha=\id$; replacing $\powf$ with $\pow^+$ (or with $\powf^+$)
  yields trace equivalence on serial (and finitely branching) LTS.

\item\label{item:sem-prob} Probabilistic trace equivalence on PTS is
  the graded equivalence induced by $\M=\M_{\CalD_f}(\A)$
  (cf.~\cref{E:graded-monad}.\ref{item:T-traces}) and
  $\alpha = \id$: The maps~$\gamma^{(k)}$ equip states with
  distributions on length-$k$ action words, and the induced
  equivalence identifies states $x$ and $y$ whenever these
  distributions coincide at~$x$ and~$y$ for all~$k$.

\item\label{item:sem-sim}
Simulation equivalence on LTS can also be
construed as a graded equivalence by taking
$\M$ to be the graded monad described in
\cref{E:graded-theory}.\ref{item:simulation-a},
and
  \begin{eqnarray*}
    \alpha_X\colon\powf(\A\times
    X) & \to & \pow^{\downarrow}_f(\A\times\powf X) \\
    S\ & \mapsto &
    {\downarrow}\{(a,\{x\})\mid (a,x)\in S\}
  \end{eqnarray*}
where $\downarrow$ takes downsets.
\end{enumerate}
\end{expl}

\begin{rem}
  It follows from the depth-1 presentations described
  in~\cref{E:graded-theory} that the graded semantics mentioned in
  \cref{E:semantics} are depth-1.
\end{rem}

\section{Graded Algebras}\label{sec:algebras}
Graded monads come equipped with graded analogues
of both the Eilenberg-Moore and Kleisli constructions for
ordinary monads. In particular, we have a
notion of \emph{graded algebra}~\cite{FKM16, MPS15}:

\begin{defn}[Graded algebra]\label{D:gradedalgebra}
Let $k\in\omega$ and let $\M$ be a graded
monad on a category $\CatC$. An \emph{$M_k$-algebra}
$A$ consists of a family of $\CatC$-objects $(A_n)_{n\leq k}$
(the \emph{carriers}) and a family of $\CatC$-morphisms%
\begin{equation}
  a^{n,m}\colon M_nA_m\to A_{n+m}		\tag{$n+m\le k$}
\end{equation}
(the \emph{structure}) such that $a^{0,n}\cdot\eta_{A_n} = \id_{A_n}$
($n\leq k$) and
\begin{equation}\label{Diagram:gradedalgebralaw}
  \begin{tikzcd}[column sep = 35]
    M_nM_mA_r \arrow[d, "\mu^{n,m}"'] \arrow[r, "M_na^{m,r}"]
    &
    M_nA_{m+r} \arrow[d, "a^{n,m+r}"] \\
    M_{n+m}A_{r} \arrow[r, "a^{n+m,r}"]
    &
    A_{n+m+r}
  \end{tikzcd}
\end{equation}
for all $n,m,r\in\omega$ such that $n+m+r\leq k$. The \emph{i-part} of
an $M_k$-algebra $A$ is the $M_0$-algebra $(A_i, a^{0,i})$.

A \emph{homomorphism} from $A$ to an $M_k$-algebra
$B$ is a family of $\CatC$-morphisms
$h_n\colon A_n\to B_n$ ($n\leq k$) such that
\[
  h_{n+m} \cdot a^{n,m} = b^{n,m}\cdot M_nh_m
  \quad\text{for all $n,m\in\omega$ s.th.~$n+m\le k$.}
\]
  We
write $\Alg_k(\M)$ for the category of $M_k$-algebras
and their homomorphisms.

We define \emph{$M_{\omega}$-algebras} (and their
homomorphisms) similarly, by allowing the indices $n,m,r$
to range over $\omega$.

\end{defn}

\begin{rem}
  The above notion of $M_{\omega}$-algebra corresponds with the
  concept of graded Eilenberg-Moore algebras introduced by Fujii et
  al.~\cite{FKM16}. Intuitively, $M_{\omega}$-algebras are devices for
  interpreting terms of unbounded uniform depth. We understand
  $M_k$-algebras~\cite{MPS15} as a refinement of $M_{\omega}$-algebras
  which allows the interpretation of terms of uniform depth \emph{at
  most~$k$}. Thus, $M_k$-algebras serve as a formalism for
  specifying properties of states exhibited in \emph{$k$ steps}. For
  example, $M_1$-algebras are used to interpret one-step modalities of
  characteristic logics for graded semantics~\cite{DMS19,FMS21a}.
  Moreover, for a depth-1 graded monad, its $M_\omega$-algebras may
  be understood as compatible chains of $M_1$-algebras~\cite{MPS15},
  and  
  a depth-1 graded monad can be reconstructed from its $M_1$-algebras.
  %
\end{rem}

\noindent We will be chiefly interested in $M_0$- and $M_1$-algebras:

\begin{expl}
Let $\M$ be a graded monad on $\Set$.
\begin{enumerate}
\item An $M_0$-algebra is just an Eilenberg-Moore algebra for the
  monad $(M_0, \eta, \mu^{0,0})$.  It follows that $\Alg_0(\M)$ is
  complete and cocomplete, in particular has coequalizers.

\item An $M_1$-algebra is a pair $((A_0, a^{0,0}), (A_1, a^{0,1}))$ of
  $M_0$-algebras -- often we just write the carriers $A_i$ to also
  denote the algebras, by abuse of notation -- equipped with a
  \emph{main structure map} $a^{1,0}\colon M_1A_0\to A_1$ satisfying
  two instances of~\eqref{Diagram:gradedalgebralaw}.  One instance
  states that $a^{1,0}$ is an $M_0$-algebra homomorphism from
  $(M_1A_0, \mu^{0,1}_A)$ to $(A_1, a^{0,1})$ (\emph{homomorphy}); the
  other expresses that $a^{1,0}\cdot \mu^{1,0}= a^{1,0}\cdot
  M_1a^{0,0}$ (\emph{coequalization}):
  \begin{equation}\label{Diagram:coequalization}
    \begin{tikzcd}[column sep = 35]
      M_1M_0A_0 \arrow[r, "\mu^{1,0}", shift left] \arrow[r,
      "M_1a^{0,0}"', shift right] & M_1A_0 \arrow[r, "a^{1,0}"] & A_1.
    \end{tikzcd}
  \end{equation}
\end{enumerate}
\end{expl}
\begin{rem}
  The free $M_n$-algebra on a set~$X$ is formed in the expected way,
  in particular has
  carriers~$M_0X,\dots,M_nX$, see~\cite[Prop.~6.3]{MPS15}.
\end{rem}



\paragraph*{Canonical algebras.}
We are going to review the basic definitions and results on
\emph{canonical $M_1$-algebras}~\cite{DMS19}. Fix a graded
monad $\M$ on $\Set$.

We write $(-)_i\colon\Alg_1(\M)\to\Alg_0(\M)$, $i = 0,1$,
for the functor which sends an $M_1$-algebra $A$ to its $i$-part
$A_i$ and sends a homomorphism $h\colon A\to B$ to
$h_i\colon A_i\to B_i$.

\begin{defn}\label{D:canonical}
An $M_1$-algebra~$A$ is \emph{canonical} if it is free over its
$0$-part with respect to $(-)_0\colon\Alg_1(\M)\to\Alg_0(\M)$.
\end{defn}

\begin{rem}\label{rem:canonical}
  The universal property of a canonical algebra~$A$ is the following:
  for every $M_1$-algebra $B$ and every $M_0$-algebra homomorphism
  $h\colon A_0\to B_0$, there exists a unique $M_1$-algebra
  homomorphism $h^\#\colon A\to B$ such that $(h^{\#})_0 = h_0$.
\end{rem}

\begin{notheorembrackets}
\begin{lem}[{\cite[Lem.~5.3]{DMS19}}]\label{L:canonical-algebra}
An $M_1$-algebra $A$ is canonical if and only if
  (\ref{Diagram:coequalization}) is a coequalizer in $\Alg_0(\M)$.
\end{lem}
\end{notheorembrackets}

\begin{expl}\label{E:canonical}
  Let $X$ be a set and let $\M$ be a depth-1 graded monad on $\Set.$
  For each ${k\in\omega}$, we may view $M_kX$ as an $M_0$-algebra with
  structure $\mu^{0,k}$. For the $M_1$-algebra $(M_kX, M_{k+1}X)$
  (with main structure map $\mu^{1,k}$), the instance of
  Diagram~(\ref{Diagram:coequalization}) required
  by~\cref{L:canonical-algebra} is a coequalizer
  by~\cref{L:depth-1}; that is, $(M_kX, M_{k+1}X, \mu^{1,k})$ is
  canonical.
\end{expl}

\section{Pre-Determinization in Eilenberg-Moore}\label{sec:determinization}
We describe a generic notion of pre-determinization (the terminology
will be explained in \cref{rem:determinization}) for coalgebras of
an endofunctor $G$ on $\Set$ with respect to a given depth-1 graded
semantics $(\alpha,\M)$, generalizing the Eilenberg-Moore-style
coalgebraic determinization construction by Silva et
al.~\cite{BBSR13}. The behavioural equivalence game introduced in the
next section will effectively be played on the pre-determinization of
the given coalgebra. We will occasionally gloss over issues of finite
branching in the examples.

We first note that every $M_0$-algebra $A$ extends (uniquely) to a
canonical $M_1$-algebra $EA$ (with $0$-part $A$), whose $1$-part and
main structure are obtained by taking the coequalizer of the pair of
morphisms in \eqref{Diagram:coequalization} (canonicity then follows
by~\cref{L:canonical-algebra}). This construction forms the object
part of a functor $\Alg_0(\M)\to\Alg_1(\M)$ which sends a homomorphism
$h\colon A\to B$ to its unique extension $Eh:=h^\sharp\colon EA\to EB$
(cf.~\cref{rem:canonical}). We write $\barM$ for the endofunctor on
$\Alg_0(\M)$ given by
\begin{equation}\label{eq:barM}
\barM:= (\Alg_0(\M)\xra{E}\Alg_1(\M)\xra{(-)_1}\Alg_0(\M)),
\end{equation}
where $(-)_1$ is the functor taking $1$-parts.  Thus, for an
$M_0$-algebra~$A_0$, $\barM(A_0)$ is the vertex of the coequalizer
\eqref{Diagram:coequalization}.

By \cref{E:canonical}, we have
\begin{equation}
  \label{eq:barM-can}
  \barM(M_kX,\mu^{0,k}_X)=(M_{k+1}X,\mu^{0,k+1}_X)
\end{equation}
for every set~$X$ and every $k\in\omega$. In particular,
\begin{equation}\label{Eq:determinization}
  U\barM F = M_1
\end{equation}
where $F\dashv U\colon \Alg_0(\M) \to \Set$ is the canonical
adjunction of the Eilenberg-Moore category of~$M_0$ -- that is,~$U$ is
the forgetful functor, and~$F$ takes free $M_0$-algebras, so
$FX=(M_0X,\mu^{00}_X)$. For an $M_1$-coalgebra
$f\colon X\to M_1X=U\barM FX$, we therefore obtain a homomorphism
$f^\#\colon FX\to\barM FX$ (in $\Alg_0(\M)$) via adjoint
transposition. This leads to the following pre-determinization
construction:

\begin{defn}\label{D:determinization}
Let $(\alpha, \M)$ be a depth-1 graded semantics on
$G$-coalgebras.
The \emph{pre-determinization} of a $G$-coalgebra
$(X, \gamma)$ under $(\alpha, \M)$
is the $\barM$-co\-al\-gebra
\begin{equation}\label{eq:det}
  (\alpha_X\cdot\gamma)^\#\colon FX\to \barM FX.
\end{equation}
\end{defn}
\begin{rem}\label{rem:determinization}
  \begin{enumerate}
  \item\label{item:predet-det} We call this construction a
    \emph{pre-}de\-ter\-mi\-ni\-zat\-ion because it will serve as a
    \emph{determinization} -- in the expected sense that the
    underlying graded equivalence transforms into behavioural
    equivalence on the determinization -- only under additional
    conditions. Notice that given a $G$-coalgebra $(X,\gamma)$,
    (finite-depth) behavioural equivalence on the $\barM$-coalgebra
    $(\alpha_X\cdot\gamma)^*_0$ is given by the canonical cone into
    the final chain
    \begin{equation*}
      1 \xla{!} \barM 1 \xla{\overbar M_1 !} \barM^{2}1 \xla{\overbar
          M_1^2 !} \cdots
    \end{equation*}
    while graded behavioural equivalence on $(X,\gamma)$ is given by
    the maps $\gamma^{(k)}$ into the sequence $M_01,M_11,M_21,\dots$,
    equivalently given as homomorphisms
    $(\gamma^{(k)})^*_0\colon FX\to(M_k1,\mu^{0,k}_1)$, whose codomains
    can, by~\eqref{eq:barM-can}, be written as the sequence
    \begin{equation*}
      F1,\quad \barM 1,\quad \barM^21,\quad \ldots
    \end{equation*}
    of $M_0$-algebras. The two sequences coincide in case $M_01=1$,
    and indeed one easily verifies that in this case, finite-depth
    behavioural equivalence on $\barM$-coalgebras coincides with
    $(\alpha,\M)$-behavioural equivalence.  For instance, this holds
    in the case of probabilistic trace equivalence
    (\cref{E:semantics}.\ref{item:sem-prob}), where $M_0=\CalD$, so
    $M_01=1$. In the case of trace equivalence
    (\cref{E:semantics}.\ref{item:sem-trace}), $M_01=1$ can be
    ensured by restricting to serial labelled transition systems,
    which, as noted in \cref{E:coalg}.\ref{E:coalg:1}, are
    coalgebras for $\pow^+(\A\times -)$ with~$\pow^+$ denoting
    non-empty powerset, so that in the corresponding variant of the
    graded monad for trace semantics, we have $M_0=\pow^+$ and hence
    $M_01=1$.

    On the other hand, the condition $M_01=1$ fails for trace
    equivalence of unrestricted systems where we have~\mbox{$M_0=\pow$,}
    which in fact
    constitutes a radical example where behavioural equivalence on the
    pre-determization is strictly coarser than the given graded
    equivalence. In this case, since the actions preserve the
    bottom~$0$, we in fact have~$\barM 1=1$: it follows that \emph{all}
    states in $\barM$-coalgebras are behaviourally equivalent (as
    the unique coalgebra structure on~$1$ is final).


  \item Using~\eqref{Eq:determinization}, we see that the underlying
    map of the pre-determinization of a coalgebra $(X,\gamma)$ is
    $(\alpha_X \cdot \gamma)^*_0 \colon M_0X \to M_1 X = U_0\barM F_0
    X$ (written using graded Kleisli star as
    per~\cref{N:star}). Indeed, one easily shows that
    $(\alpha_X \cdot\gamma)^*_0$ is an $M_0$-algebra
    morphism 
    $(M_0 X, \mu^{0,0}_X)\to\barM
    (M_0X,\mu^{0,0}_X)=(M_1X,\mu^{0,1}_X)$ satisfying
    $(\alpha_X \cdot \gamma)^*_0 \cdot \eta_X = \alpha_X \cdot
    \gamma$. Thus, it is the adjoint transpose in~\eqref{eq:det}.

  \item As indicated above, pre-de\-ter\-mi\-ni\-za\-tion captures the
    Eilenberg-Moore style generalized determinization by Silva et
    al.~\cite{BBSR13} as an instance. Indeed, for a monad $T$ and an
    endofunctor $F$, both on the category $\CatC$, one considers a
    coalgebra $\gamma\colon X \to FTX$. Assuming that $FTX$ carries
    the structure of an Eilenberg-Moore algebra for $T$ (e.g.~because
    the functor~$F$ lifts to the category of Eilenberg-Moore algebras
    for $T$), one obtains an
    $F$-coalgebra $\gamma^\sharp\colon TX \to FTX$ by taking the
    unique homomorphic extension of $\gamma$. Among the concrete
    instances of this construction are the well-known powerset
    construction of non-deterministic automata (take $T = \pow$ and
    $F = 2 \times (-)^A$), the non-determinization of alternating
    automata and that of Markov decision processes~\cite{JSS15}.

    To view this as an instance of pre-de\-ter\-mi\-ni\-za\-tion, take
    the graded monad with $M_n = F^nT$
    (\cref{E:graded-monad}.\ref{E:graded-monad:4}), let $G = FT$,
    and let $\alpha = \id_{FT}$. Using~\eqref{Eq:determinization}, we
    see that $(\alpha_X \cdot \gamma)^\#$ in~\eqref{eq:det} is the
    generalized determinization $\gamma^\sharp$ above.
  \item We emphasize that the construction applies completely
    universally; e.g.~we obtain as one instance a `determinization' of
    serial labelled transition systems modulo similarity, which
    transforms a coalgebra $X\to\pow^+(\A\times X)$ into an
    $\barM$-coalgebra
    $\pow^+(X)\to\pow^{\downarrow}(\A\times \pow^+(X))$
    (\cref{E:graded-theory}.\ref{item:simulation-a}); instantiating
    the observations in item~\ref{item:predet-det}, we obtain that
    finite-depth behavioural equivalence of $\barM$-coalgebras (see
    \cref{expl:barM} for the description of~$\barM$) coincides with
    finite-depth mutual similarity.
  \end{enumerate}
\end{rem}

\begin{expl}\label{expl:barM}
  We give a description of the functor~$\barM$ on $M_0$-algebras
  constructed above in some of the running examples.
  \begin{enumerate}[wide]
  \item For graded monads of the form $\M_G$, which capture
    finite-depth behavioural equivalence
    (\cref{E:graded-monad}.\ref{E:graded-monad:1}), we have
    $M_0=\Id$, so $M_0$-algebras are just sets, and under this
    correspondence, $\barM$ is the original functor~$G$.

  \item\label{item:mono-traces} Trace semantics of LTS
    (\cref{E:graded-theory}.\ref{item:traces-a}): Distribution of
    actions over the join semilattice operations ensures that depth-1
    terms over a join semilattice~$X$ can be normalized to sums of the
    form $\sum_{a\in \A}a(x_a)$, with $x_a\in X$ (possibly
    $x_a=0$). It follows that $\barM$ is simply given by
    $\barM X=X^\A$ ($\A$-th power, where~$\A$ is the finite set of labels).
    Other forms of trace semantics are treated similarly. 
  \item In the graded theory for simulation
    (\cref{E:graded-theory}.\ref{item:simulation-a}), the
    description of the induced graded monad~\cite{DMS19} extends
    analogously to~$\barM$, yielding that $\barM B$ is the join
    semilattice of finitely generated downwards closed subsets of
    $\A\times B$ where, again,~$\A$ carries the discrete
    ordering.
  \end{enumerate}
\end{expl}

\begin{rem}
  The assignment $\M \mapsto \barM$ exhibits the category $\K$ of
  depth-1 graded monads whose $0$-part is the monad
  $(M_0, \eta, \mu^{0,0})$ as a coreflective subcategory (up to
  isomorphism) of the category $\Fun(\Set^{M_0})$ of all endofunctors
  on the Eilenberg-Moore category of that monad.

  Indeed, given an endofunctor $H$ on $\Set^{M_0}$ we form the
  $6$-tuple
    $(M_0, UHF, \eta, \mu^{0,0}, \mu^{0,1},\mu^{1,0}),$
  where the latter two natural transformations arise from the counit
  $\eps\colon FU \to \Id$ of the canonical adjunction $F\dashv U\colon\Alg_0(\M_0)\to\Set$:
  \begin{align*}
    \mu^{0,1} &= (M_0UHF = UFUHF \xra{U\eps HF} UHF\big);\\
    \mu^{1,0} &= (UHFM_0 = UHFUF \xra{UHF\eps F} UHF\big).
  \end{align*}
  It is not difficult to check that this data satisfies all applicable
  instances of the graded monad laws. Hence, it specifies a
  depth-1 graded monad $R(H)$~\cite[Thm.~3.7]{DMS19}; this
  assignment is the object part of a functor $R\colon\Fun(\Set^{M_0})\to\K$.

  In the other direction, we have for each depth-1 graded monad $\M$
  with $0$-part $M_0$ the endofunctor $I(\M) =
  \barM$. By~\eqref{Eq:determinization}, we have $RI(\M) = \M$.
  Now, given a depth-1 graded monad $\M$ and an endofunctor $H$ on
  $\Set^{M_0}$, consider $\barM = IR(H)$ (so that $M_1 = UHF$). We
  obtain for every algebra $(A,a)$ in $\Set^{M_0}$ a homomorphism
  $c_{(A,a)}\colon \barM (A,a) \to H(A,a)$ by using the coequalizer
  defining $\barM(A,a)$ (cf.~\cref{L:canonical-algebra}):
  \[
    \begin{tikzcd}
      M_1M_0 A
      \ar[yshift=2]{r}{\mu^{0,1}_A}
      \ar[yshift=-2]{r}[swap]{M_1a}
      &
      M_1A=HFA
      \ar[->>]{r}
      \ar{rd}{Ha}
      &
      \barM (A,a)
      \ar[dashed]{d}{c_{(A,a)}}
      \\
      &&
      H(A,a)
    \end{tikzcd}
  \]
  Note that $M_1M_0A$ is the carrier of the Eilenberg-Moore algebra
  $HFA = H(M_0A, \mu^{0,0}_A)$ and similarly for the middle object (in
  both cases we have omitted the algebra structures given by
  $\mu^{0,1}_{M_0 A}$ and $\mu^{0,1}_A$ coming from the graded monad
  $I(H)$).  It is easy to see that the homomorphism $Ha$ merges the
  parallel pair, and therefore we obtain the dashed morphism such that
  the triangle commutes, yielding the components of a natural
  transformation $c\colon \barM \to H$ which is couniversal: for each
  depth-1 graded monad $\mathbb N$ whose $0$-part is $M_0$ and each
  natural transformation $h\colon \barM \to H$, there is a unique
  natural transformation $m_1\colon N_1 \to M_1 = UHF$ such that
  $m = (id_{M_0}, m_1)$ is a morphism of graded monads from
  $\mathbb N$ to~$\M$ and $c \cdot I(m) = h$. This shows that $I
  \dashv R$.
\end{rem}

\section{Behavioural Equivalence Games}\label{S:games}
Let $\CalS = (\alpha, \M)$ be a depth-1 graded semantics for an
endofunctor $G$ on $\Set$. We are going to describe a game for playing
out depth-$n$ behavioural equivalence under $\CalS$-semantics on
states in $G$-coalgebras.

We first give a description of the game in the syntactic language of
graded equational reasoning, and then present a more abstract
categorical definition. Given a coalgebra $(X,\gamma)$, we will see
the states in~$X$ as variables, and the map $\alpha_X\cdot\gamma$ as
assigning to each variable~$x$ a depth-1 term over~$X$; we can regard
this assignment as a (uniform-depth) substitution~$\sigma$. A
configuration of the game is a pair of depth-0 terms over~$X$; to play
out the equivalence of states $x,y\in X$, the game is started from the
initial configuration $(x,y)$.  Each round of the game then proceeds
in two steps: First, Duplicator plays a set~$Z$ of equalities between
depth-0 terms over~$X$ that she claims to hold under the
semantics. This move is admissible in the configuration $(s,t)$
if~$Z\vdash s\sigma=t\sigma$. Then, Spoiler challenges one of the
equalities claimed by Duplicator, i.e.~picks an
element~$(s',t')\in Z$, which then becomes the next configuration. Any
player who cannot move, loses.  After~$n$ rounds have been played,
reaching the final configuration $(s,t)$, Duplicator wins if
$s\theta =t\theta$ is a valid equality, where~$\theta$ is a
substitution that identifies all variables. We refer to this last
check as \emph{calling the bluff}. Thus, the game plays out an
equational proof between terms obtained by unfolding depth-0 terms
according to~$\sigma$, cutting off after~$n$ steps.

We introduce some technical notation to capture the notion of
admissibility of~$Z$ abstractly:
\begin{notn}\label{N:admissible}
  Let $Z\subseteq M_0X\times M_0X$ be a relation, and let
  $c_Z\colon M_0X\to C_Z$ be the coequalizer in $\Alg_0(\M)$ of the
  homomorphisms $ \ell_0^*, r_0^*\colon M_0Z\to M_0X $ given by
  applying the Kleisli star~\eqref{Eqn:Kleisli-star} to the
  projections $\ell, r\colon Z\to M_0X$. We define a homomorphism
  $\barZ\colon M_0X\to M_1C_Z$ in $\Alg_0(\M)$ by
  \begin{equation}\label{eq:barZ}
    \barZ = \big(M_0X\xra{(\alpha_X\cdot\gamma)^*_0} M_1X = \barM M_0X
    \xra{\overbar M_1 c_Z} \barM C_Z\big)
  \end{equation}
  (omitting algebra structures, and again using the Kleisli star).
\end{notn}

\begin{rem}\label{R:coeq}
  Using designators as in \cref{N:admissible}, we note:
  \begin{enumerate}
  \item\label{R:coeq:1}\label{R:coeq:2} By the universal property of
    $\eta_Z\colon Z \to M_0Z$, an $M_0$-algebra homomorphism
    $h\colon M_0X \to A$ merges $\ell, r$ iff it merges
    $\ell^*_0, r^*_0$. This implies that the coequalizer $M_0X \xra{c_Z}C_Z$
    quotients the free $M_0$-algebra $M_0X$ by the congruence
    generated by~$Z$. Also, it follows that in case~$Z$ is already an
    $M_0$-algebra and $\ell, r\colon Z \to M_0X$ are $M_0$-algebra
    homomorphisms (e.g.~when $Z$ is a congruence), one may take
    $c_Z\colon M_0X \to C_Z$ to be the coequalizer of $\ell, r$.

  \item\label{item:barZ} The map $\barZ\colon M_0X\to\barM C_Z$
    associated to the relation $Z$ on $M_0X$ may be understood as
    follows. As per the discussion above, we view the states of the
    coalgebra $(X,\gamma)$ as variables, and the map
    $X\xra{\gamma} GX\xra{\alpha_X} M_1X$ as a substitution mapping a
    state $x \in X$ to the equivalence class of depth-1 terms encoding
    the successor structure $\gamma(x)$. The second factor $\barM c_Z$
    in~\eqref{eq:barZ} then essentially applies the relations given by
    the closure of $Z$ under congruence w.r.t.~depth-0 operations,
    embodied in~$c_Z$ as per~\ref{R:coeq:1}, under depth-1 operations
    in (equivalence classes of) of depth-1 terms in $M_1X$; to sum up,
    $\barM c_Z$ merges a pair of equivalence classes $[t], [t']$ iff
    $Z\vdash t=t'$ in a depth-1 theory presenting $\M$ (in notation as
    per \cref{sec:prelims}). 
  \end{enumerate}
\end{rem}

\begin{defn}\label{def:game}
  For $n\in\omega$, the \emph{$n$-round $\CalS$-behavioural
    equivalence game} $\CalG_n(\gamma)$ on a $G$-coalgebra
  $(X, \gamma)$ is played by Duplicator (D) and Spoiler
  (S). \emph{Configurations} of the game are pairs
  $(s,t)\in M_0(X)\times M_0(X)$. Starting from an \emph{initial
    configuration} designated as needed, the game is played for~$n$
  rounds. Each round proceeds in two steps, from the current
  configuration~$(s,t)$: First, D chooses a relation
  $Z\subseteq M_0X\times M_0X$ such that $\barZ(s) = \barZ(t)$
  (for~$\barZ$ as per \cref{N:admissible}). Then,~S
  picks an element~$(s',t') \in Z$, which becomes the next configuration. Any
  player who cannot move at his turn, loses. After~$n$ rounds have
  been played,~D wins if $M_0!(s_n) = M_0!(t_n)$; otherwise,~S wins.
\end{defn}

\begin{rem}
  By the description of~$\barZ$ given in
  \cref{R:coeq}.\ref{item:barZ}, the categorical definition of the
  game corresponds to the algebraic one given in the lead-in
  discussion. The final check whether $M_0!(s_n)=M_0!(t_n)$
  corresponds to what we termed \emph{calling the bluff}. The apparent
  difference between playing either on depth-0 terms or on elements
  of~$M_0X$, i.e.~depth-0 terms modulo derivable equality, is absorbed
  by equational reasoning from~$Z$, which may incorporate also the
  application of depth-0 equations.
\end{rem}
\begin{rem}
  A pair of states coming from different coalgebras $(X,\gamma)$ and
  $(Y,\delta)$ can be treated by considering those states as elements
  of the coproduct of the two coalgebras:
  \[
    X+Y \xra{\gamma + \delta} GX + GY \xra{[G\inl, G\inr]} G(X+Y),
  \]
  where $X \xra{\inl} X+Y \xla{\inr} Y$ denote the coproduct
  injections. There is an evident variant of the game played on two
  different coalgebras $(X,\gamma)$, $(Y,\delta)$, where moves of~D
  are subsets of $M_0X\times M_0Y$. However, completeness of this
  version depends on additional assumptions on~$\M$, to be clarified
  in future work. For instance, if we instantiate the graded monad for
  traces with effects specified by~$T$
  (\cref{E:graded-monad}.\ref{item:T-traces}) to~$T$ being the free
  real vector space monad, and a state~$x\in X$ has successor
  structure $2\cdot x'-2\cdot x''$, then~D can support equivalence
  between~$x$ and a deadlock~$y\in Y$ (with successor structure~$0$)
  by claiming that $x'=x''$, but not by any equality between terms
  over~$X$ with terms over~$Y$. That is, in this instance, the variant
  of the game where~D plays relations on $M_0X\times M_0Y$ is not
  complete.
\end{rem}

\noindent Soundness and completeness of the game with respect to
$\CalS$-behavioural equivalence is stated as follows.

\begin{thm}\label{T:sound-complete}
  Let $(\alpha, \M)$ be a depth-1 graded semantics for a functor~$G$
  such that $\barM$ preserves monomorphisms, and let $(X, \gamma)$ be
  a $G$-coalgebra. Then, for all $n\in\omega$,~D wins $(s, t)$ in
  $\CalG_n(\gamma)$ if and only if
  $(\gamma^{(n)})^*_0(s) = (\gamma^{(n)})^*_0(t)$.
\end{thm}

\begin{cor}
  States $x,y$ in a $G$-coalgebra $(X, \gamma)$ are
  $\CalS$-behaviourally equivalent if and only if~D wins
  $(\eta(x), \eta(y))$ for all $n\in\omega$.
\end{cor}

\begin{rem}\label{rem:monos}
  In algebraic terms, the condition that~$\barM$ preserves
  monomorphisms amounts to the following: In the derivation of an
  equality of depth-1 terms~$s,t$ over~$X$ from depth-0 relations
  over~$X$ (i.e.~from a presentation of an $M_0$-algebra by relations
  on generators~$X$), if~$X$ is included in a larger set~$Y$ of
  variables with relations that conservatively extend those on~$X$,
  i.e.~do not imply additional relations on~$X$, then it does not
  matter whether the derivation is conducted over~$X$ or more
  liberally over~$Y$. Intuitively, this property is needed because not
  all possible $n$-step behaviours, i.e.~elements of~$Y=M_n1$, are
  realized by states in a given coalgebra on~$X$. Preservation of
  monos by~$\barM$ is automatic for graded monads of the form $\M_G$
  (\cref{E:graded-monad}.\ref{E:graded-monad:1}), since $M_0=\Id$
  in this case. In the other running examples, preservation of monos
  is by the respective descriptions of~$\barM$ given in
  \cref{expl:barM}.
\end{rem}
\begin{expl}\label{expl:bisim-instance}
  We take a brief look at the instance of the generic game for the
  case of bisimilarity on finitely branching LTS (more extensive
  examples are in \cref{sec:cases}), i.e.~we consider the depth-1
  graded semantics $(\id, \M_G)$ for the functor
  $G=\powf(\A\times(-))$. In this case, $M_0=\Id$, so when playing on
  a coalgebra $(X,\gamma)$,~D plays relations~$Z\subseteq X\times
  X$. If the successor structures of states~$x,y$ are represented by
  depth-1 terms $\sum_{i}a_i(x_i)$ and $\sum_j b_j(y_j)$,
  respectively, in the theory $\JSL(\A)$
  (\cref{E:graded-theory}.\ref{item:jsl-a}), then~D is allowed to
  play~$Z$ iff the equality $\sum_{i}a_i(x_i)=\sum_j b_j(y_j)$ is
  entailed by~$Z$ in $\JSL(\A)$. This, in turn, holds iff for
  each~$i$, there is~$j$ such that $a_i=b_j$ and $(x_i,y_j)\in Z$, and
  symmetrically. Thus~$Z$ may be seen as a pre-announced
  non-deterministic winning strategy for~D in the usual bisimilarity
  game where~S moves first (\cref{sec:prelims}):~D announces that
  if~S moves from, say,~$x$ to~$x_i$, then she will respond with
  some~$y_j$ such that $a_i=b_j$ and $(x_i,y_j)\in Z$.
\end{expl}

\section{Infinite-depth behavioural
  equivalence}\label{sec:infinte-depth}

\sloppypar
\noindent We have seen in \cref{sec:determinization} that in case
\mbox{$M_01=1$}, $(\alpha, \M)$-behavioural equivalence
on~$G$-coalgebras coincides, via a determization construction, with
finite-depth behavioural equivalence on $\barM$-coalgebras for a
functor $\barM$ on $M_0$-algebras constructed from~$\M$. If~$G$ is
finitary, then finite-depth behavioural equivalence coincides with
full behavioural equivalence~(\cref{rem:finite-depth}), but in
general, finite-depth behavioural equivalence is strictly
coarser. Previous treatments of graded semantics stopped at this
point, in the sense that for non-finitary functors (which describe
infinitely branching systems), they did not offer a handle on
infinite-depth equivalences such as full bisimilarity.
In case $M_01=1$, a candidate for a notion
of infinite-depth equivalence induced by a graded semantics arises via
full behavioural equivalence of $\barM$-coalgebras. We fix this notion
explicitly:
\begin{defn}
  States $x,y$ in a $G$-coalgebra $(X,\gamma)$ are
  \emph{in\-fin\-ite-depth $(\alpha,\M$)-behaviourally equivalent}
  if~$\eta(x)$ and~$\eta(y)$ are behaviourally equivalent in the
  pre-det\-er\-min\-i\-za\-tion of~$(X,\gamma)$ as described in
  \cref{S:games}.
\end{defn}
\noindent We hasten to re-emphasize that this notion in general only
makes sense in case $M_01=1$. We proceed to show that infinite-depth
equivalence is in fact captured by an infinite variant of the
behavioural equivalence game of \cref{S:games}.

Since infinite depth-equivalences differ from finite-depth ones only
in settings with infinite branching, we do not assume in this section
that~$G$ or~$\M$ are finitary, and correspondingly work with
generalized graded theories where operations may have infinite
arities~\cite{MPS15}; we assume arities to be cardinal numbers. We
continue to be interested only in depth-1 graded monads and theories,
and we fix such a graded monad~$\M$ and associated graded theory for the
rest of this section. The notion of derivation is essentially the same
as in the finitary case, the most notable difference being that the
congruence rule is now infinitary, as it has one premise for each
argument position of a given possibly infinitary operator. We do not
impose any cardinal bound on the arity of operations; if all
operations have arity less than~$\kappa$ for a regular
cardinal~$\kappa$, then we say that the monad is
\emph{$\kappa$-ary}.
\begin{rem}\label{R:final-coalg}
  One can show using tools from the theory of locally presentable
  categories that $\barM$ has a final coalgebra if~$\M$ is
  $\kappa$-ary in the above sense. To see this, first note that
  $\Alg_0(\M)$ is locally $\kappa$-presentable if $M_0$ is
  $\kappa$-accessible~\cite[Remark~2.78]{AR94}. Using a somewhat
  similar argument one can prove that $\Alg_1(\M)$ is also locally
  $\kappa$-presentable. Moreover, the functor $\barM$ is
  $\kappa$-accessible, being the composite~\eqref{eq:barM} of the left
  adjoint $E\colon \Alg_0(\M) \to \Alg_1(\M)$ (which preserves all
  colimits) and the $1$-part functor
  $(-)_1\colon \Alg_1(\M) \to \Alg_0(\M)$, which preserves
  $\kappa$-filtered colimits since those are formed componentwise.
  It follows that $\barM$ has a final
  coalgebra~\cite[Exercise~2j]{AR94}. Alternatively, existence of a
  final $\barM$-coalgebra will follow from \cref{thm:fin-coalg}
  below. %
\end{rem}

\noindent
Like before, we \emph{assume that $\barM$ preserves monomorphisms}.

\begin{expl}
  We continue to use largely the same example theories as in
  \cref{E:graded-theory}, except that we allow operations to be
  infinitary. For instance, the \emph{graded theory of complete join
    semilattices over~$\A$} has as depth-1 operations all formal sums
  $\sum_{i\in I}a_i(-)$ where~$I$ is now some (possibly infinite)
  index set; the axioms are then given in the same way as in
  \cref{E:graded-theory}.\ref{item:jsl-a}, and all depth-1
  equations
  \[
    \textstyle \sum_{i\in I}a_i(x) = \sum_{j\in J} b_j(y)
  \]
  such that  $\{(a_i, x_i)\mid i\in I\}=\{(b_j, y_j)\mid j\in  J\}$.
  This theory presents the graded monad~$\M_G$ for
  $G=\pow(\A\times(-))$.  
\end{expl}
%

\noindent The infinite game may then be seen as defining a notion of
derivable equality on infinite-depth terms by playing out a
non-standard, infinite-depth equational proof; we will make this view
explicit further below. In a less explicitly syntactic version, the
game is defined as follows.
\begin{defn}[Infinite behavioural equivalence game]
  The \emph{infinite $(\alpha,\M$)-behavioural equivalence game}
  $\CalG_\infty(\gamma)$ on a $G$-coalgebra~$(X,\gamma)$ is played by
  Spoiler~(S) and Duplicator~(D) in the same way as the finite
  behavioural equivalence game (\cref{def:game}) except that the
  game continues forever unless one of the players cannot move. Any
  player who cannot move, loses. Infinite matches are won by~D.
\end{defn}
\noindent As indicated above, this game captures infinite-depth
$(\alpha,\M)$-behavioural equivalence (under the running assumption
that~$\barM$ preserves monomorphisms):
\begin{thm}\label{thm:infinite-depth-games}
  Given a $G$-coalgebra~$(X,\gamma)$, two states~$s,t$ in the
  pre-determinization of~$\gamma$ are behaviourally equivalent iff D
  wins the infinite $(\alpha,\M$)-behavioural equivalence game
  $\CalG_\infty(\gamma)$ from the initial configuration $(s,t)$.
\end{thm}
\begin{cor}
  Two states $x,y$ in a $G$-coalgebra $(X,\gamma)$ are infinite-depth
  $(\alpha,\M)$-behaviourally equivalent iff D wins the infinite
  $(\alpha,\M$)-behavioural equivalence game $\CalG_\infty(\gamma)$
  from the initial configuration $(\eta(x),\eta(y))$.
\end{cor}
\begin{rem}
  Like infinite-depth $(\alpha,\M)$-behavioural equivalence, the
  infinite $(\alpha,\M$)-behavioural equivalence game is sensible only
  in case $M_01=1$. For instance, as noted in
  \cref{sec:determinization}, in the graded monad for trace
  semantics (\cref{E:graded-theory}.\ref{item:pt-a}), which does not
  satisfy this condition, behavioural equivalence of
  $\barM$-coalgebras is trivial. In terms of the game,~D wins every
  position in $\CalG_\infty(\gamma)$ by playing
  $Z=\{(t,0)\mid t\in M_0X\}$ -- since the actions preserve the bottom
  element~$0$, this is always an admissible move. In the terminology
  introduced at the beginning of \cref{S:games}, the reason that~D
  wins in this way is that in the infinite game, her bluff is never
  called ($M_0!(t)$ will in general not equal $M_0!(0)=0$).  However, see
  \cref{expl:inf-depth}.\ref{item:inf-trace} below.
\end{rem}

\begin{expl}\label{expl:inf-depth}
  \begin{enumerate}
  \item\label{item:inf-trace} As noted in
    \cref{rem:determinization}.\ref{item:predet-det}, the graded
    monad for trace semantics can be modified to satisfy the
    condition~$M_01=1$ by restricting to serial labelled transition
    systems. 
    By modification of~\cref{expl:barM}.\ref{item:mono-traces}, we
    obtain that in this setting, $\barM X$ consists of partial maps
    $\A\rightharpoonup X$ that are not everywhere undefined. It
    follows that in this case, infinite-depth
    $(\alpha,\M)$-behavioural equivalence is just finite trace
    equivalence, and thus coincides with plain
    $(\alpha,\M)$-behavioural equivalence.
  \item In the case of graded monads $\M_G$
    (\cref{E:graded-monad}.\ref{E:graded-monad:1}), which so far
    were used to capture finite-depth behavioural equivalence in the
    standard (branching-time) sense, we have $M_0=\Id$; in particular,
    $M_01=1$. In this case, the infinite-depth behavioural equivalence
    game instantiates to a game that characterizes full behavioural
    equivalence of $G$-coalgebras. Effectively, a winning strategy
    of~D in the infinite game~$\CalG_\infty(\gamma)$ on a
    $G$-coalgebra $(X,\gamma)$ amounts to a
    relation~$R\subseteq X\times X$ (the positions of~D actually
    reachable when~D follows her winning strategy) that is a
    \emph{precongruence} on~$(X,\gamma)$~\cite{AczelMendler89}.
  \end{enumerate}
\end{expl}

\begin{rem}[Fixpoint computation]
  Via its game characterization (\cref{thm:infinite-depth-games}),
  infinite-depth $(\alpha,\M)$-behavioural equivalence can be cast as
  a greatest fixpoint, specifically of the monotone function~$F$ on
  $\pow(M_0X\times M_0X)$ given by
  \begin{equation*}
    F(Z)=\{(s,t)\in M_0X\times M_0X\mid \barZ(s)=\barZ(t)\}.
  \end{equation*}
  If~$M_0$ preserves finite sets, then this fixpoint can be computed
  on a finite coalgebra $(X,\gamma)$ by fixpoint iteration; since
  $F(Z)$ is clearly always an equivalence relation, the iteration
  converges after at most $|M_0X|$ steps, e.g.~in exponentially many
  steps in case $M_0=\pow$. In case $M_0X$ is infinite (e.g.~if
  $M_0=\CalD$), then one will need to work with finite representations
  of subspaces of~$M_0X\times M_0X$. We leave a more careful analysis
  of the algorithmics and complexity of solving infinite
  $(\alpha,\M)$-behavioural equivalence games to future work. We do
  note that on finite coalgebras, we may assume w.l.o.g.~that both the
  coalgebra functor~$G$ and graded monad~$\M$ are finitary, as we can
  replace them with their finitary parts if needed (e.g.~the powerset
  functor $\pow$ and the finite powerset functor~$\powf$ have
  essentially the same finite coalgebras). If additionally~$M_01=1$,
  then~$(\alpha,\M)$-behavioural equivalence coincides with
  infinite-depth $(\alpha,\M)$-behavioural equivalence, so that we
  obtain also an algorithmic treatment of $(\alpha,\M)$-behavioural
  equivalence. By comparison, such a treatment is not immediate from
  the finite version of the game, in which the number of rounds is
  effectively chosen by Spoiler in the beginning.
\end{rem}
\noindent Assume from now on that $\M$ is $\kappa$-ary.
We note that in this case, we can describe the final $\barM$-coalgebra
in terms of a syntactic variant of the infinite game that is played on
infinite-depth terms, defined as follows.
\begin{defn}[Infinite-depth terms]
  Recall that we are assuming a graded signature~$\Sigma$ with
  operations of arity less than~$\kappa$. A \emph{(uniform)
    infinite-depth \mbox{($\Sigma$-)}term} is an infinite tree with
  ordered branching where each node is labelled with an
  operation~$f\in\Sigma$, and then has as many children as given
  by the arity of~$f$; when there is no danger of confusion, we
  will conflate nodes with (occurrences of) operations. We require
  moreover that every infinite path in the tree contains infinitely
  many depth-1 operations (finite full paths necessarily end in
  constants). We write $\Termsarg{\Sigma,\infty}$ for the set of
  infinite-depth $\Sigma$-terms. By cutting off at the top-most
  depth-1 operations, we obtain for every $t\in\Termsarg{\Sigma,\infty}$
  a \emph{top-level decomposition} $t=t_1\sigma$ into a depth-1 term
  $t_1\in\Termsarg{\Sigma,1}(X)$, for some set~$X$, and a substitution
  $\sigma\colon X\to\Termsarg{\Sigma,\infty}$.
\end{defn}

\begin{defn}
  The \emph{syntactic infinite $(\alpha,\M$)-behavioural equivalence
    game} $\CalG^\syn_\infty$ is played by~S and~D. Configurations of
  the game are pairs $(s,t)$ of infinite-depth $\Sigma$-terms. For
  such $(s,t)$, we can assume, by the running assumption that~$\barM$
  preserves monomorphisms, that the top level decompositions
  $s=s_1\sigma$, $t=t_1\sigma$ are such that
  $s_1,t_1\in\Termsarg{\Sigma,1}(X)$,
  $\sigma\colon X\to\Termsarg{\Sigma,\infty}$ for the
  same~$X,\sigma$. Starting from a designated initial configuration,
  the game proceeds in rounds. In each round, starting from a current
  such configuration~$(s,t)$,~D first chooses a
  relation~$Z\subseteq\Termsarg{\Sigma,0}(X)\times\Termsarg{\Sigma,0}(X)$
  such that $Z\vdash s_1 = t_1$ in the graded theory that
  presents~$\M$ (cf.~\cref{sec:prelims}). \mbox{Then, S} selects an
  element $(u,v)\in Z$, upon which the game reaches the new
  configuration $(u\sigma,v\sigma)$. The game proceeds forever unless
  a player cannot move. Again, any player who cannot move, loses, and
  infinite matches are won by~D. We write $s\infgameeq t$ if~D
  wins~$\CalG^\syn_\infty$ from position $(s,t)$.
\end{defn}
\noindent We construct an $\barM$-coalgebra on the set
\(
U=\Termsarg{\Sigma,\infty}/{\infgameeq}
\)
of infinite-depth terms modulo the winning region of~D as follows. We
make~$U$ into an $M_0$-algebra by letting depth-0 operations act by
term formation. We then define the coalgebra
structure~$\zeta\colon U \to \barM U$ by
\(
  \zeta(q(t_1\sigma)) = \barM((q\cdot\sigma)^*_0)([t_1])
\)
(using Kleisli star as per \cref{N:star}) where
$t_1\sigma$ is a top-level decomposition of an infinite-depth term,
with $t_1\in\Termsarg{\Sigma,1}(X)$;
\[
  [-]\colon  \Termsarg{\Sigma,1}(X)\to M_1X=\barM M_0 X
  \quad\text{and}\quad
  q\colon \Termsarg{\Sigma,\infty}\to U
\]
denote canonical quotient maps.
These data are well-defined.%
\begin{thm}\label{thm:fin-coalg}
  The coalgebra $(U,\zeta)$ is final.%
\end{thm}

\section{Case studies}\label{sec:cases} 
We have already seen (\cref{expl:bisim-instance}) how the standard
bisimilation game arises as an instance of our generic game. We
elaborate on some further examples. %
\takeout{\subsection{Trace equivalence} In this subsection, we
  illustrate the admissible moves of Duplicator in concrete terms and,
  subsequently, establishing the game-theoretic characterisation of
  trace equivalence.
\bknote{Example~2.9 does not talk about trace equivalence, we should explain the connection better.}
To this end, consider an LTS $\gamma\colon X\to\powf(\A\times X)$
and recall the depth-1 graded monad attributed to trace equivalence
from \cite{MPS15}.

%
\takeout{In particular, $\Sigma=\{\bot,\vee\} \cup \{a.\_ \mid a\in \A\}$ with action prefixing as the only depth-1 operation; $\E$ has all the axioms of join semilattice and the following depth-1 equations.
\[
a.\bot = \bot \qquad a.(x\vee y) = a.x \vee a.y
\]}
%

\takeout{We begin by stating the coequaliser $C_Z,c_Z$ for a given relation $Z$ (cf.\thinspace Notation~\ref{N:admissible}) in more concrete terms. Define $Z' \subseteq M_1 X \times M_1 X$ as follows:
\[
t \mathrel {Z'} t' \iff \exists_{s\in M_1 Z}\ l_1^*s =t \land r_1^*s=t'.
\]
Then the object $C_Z$ is given by the quotient set $M_1X/ Z''$, where $Z''$ is the least equivalence which is closed with respect to all the depth-0 operations and which includes $Z'$. In short, we say that $Z''$ is the \emph{congruence closure} of $Z'$.
\begin{propn}
  The object $C_Z$ together with its quotient map $c_Z$ indeed forms the coequaliser of the arrows $l_1^*,r_1^*$.
\end{propn}}
\begin{expl}
  Consider the following process terms (defined in a fragment of CCS
  signature):
  \[
    p_1\equiv a.p_1',\ p_2\equiv a.p_2' + b.p_2'',\ p_3\equiv b.p_3',
  \]
  where $p_1',p_2',p_2'',p_3'$ are all inaction (i.e.~the constant $\0$) and $\equiv$ denotes the syntactic equality. Clearly we find that the set $s=\{p_1,p_2\}$ and $t=\{p_2,p_3\}$ are trace equivalent.

  In particular, $s,t$ have the same traces of length 1, which we argue next through our game. Duplicator plays the relation $Z$ viewed as a set of equations between $M_0$-terms:
  \[
  Z = \{p_1' + p_2' = p_2',\ p_3' + p_2'' = p_2'' \}.
  \]
  We claim that the relation $Z$ is admissible at position $(s,t)$ because the $M_1$-terms $(\alpha\cdot\gamma)^\# s=a.p_1' + a.p_2' + b.p_2''$ and $(\alpha\cdot\gamma)^\# t=a.p_2' + b.p_2'' + b.p_3'$ are related in the congruence closure $Z''$ of $Z'$. To establish this first observe that
  \[
  (a.p_1' + a.p_2') \mathrel {Z'} a.p_2'\ \text{and}\ b.p_2''\mathrel {Z'} (b.p_3' + b.p_2'').
  \]
  Clearly $(\alpha\cdot\gamma)^\# s\mathrel {Z''} (\alpha\cdot\gamma)^\# t$. Now Spoiler can pick either of the two equations in $Z$. Moreover, both pairs $(\{p_1',p_2'\},\{p_2'\})$ and $(\{p_3',p_2''\},\{p_2''\})$ are mapped to a common point (the singleton containing empty trace) by $M_1!$; thus, resulting in two 1-round matches both won by Duplicator.
\end{expl}
\begin{propn}
  Every $M_1$-term $t$ is derivably equivalent to a term of the form $\bigvee_{a}a.t_a$, where $t_a$ is some $M_0$-term.
\end{propn}
\begin{propn}
  Suppose the normal form of $M_1$-terms $(\alpha\cdot\gamma)^*_0\Box$ is $\bigvee_{a} a.\Box_a$ (for $\Box\in\{s,t\}$). Then a relation $Z$ on $M_0X$ satisfying the following conditions for each action $a$ is admissible at $(s,t)$.
  \begin{enumerate}
    \item $\forall_{x\in s_a}\exists_{t'}\ t'\subseteq t_a \land \left((t',x\vee t') \in Z \lor (x\vee t',t')\in Z \right)$, and
    \item $\forall_{y\in t_a}\exists_{s'}\ s'\subseteq s_a \land \left((s',y\vee s') \in Z \lor (y\vee s',s')\in Z\right)$.
  \end{enumerate}
\end{propn}
\begin{proof}
  To show admissibility of $Z$ at $(s,t)$, it suffices to show that $\T+Z \vdash s_a=t_a$. Condition 1 ensures that
  \[
  \T+Z \vdash s_a \lor t_a = \bigvee_{x\in s_a} x \vee t_a = t_a.
  \]
  Likewise Condition 2 ensures that $\T+Z \vdash s_a \lor t_a = s_a$. Thus, $\T +Z \vdash s_a=t_a$.
\end{proof}
A natural question to ask is whether there exists an algorithm to
determine whether $s,t$\hbnote{I changed $U,V$ to $s,t$ inline with the notation in Definition~\ref{def:game}.} are in the congruence closure of $Z$. In fact
there are algorithms to do this for the powerset monad $\powf$
\cite{bp:checking-nfa-equiv} and for certain semiring monads
\cite{bkk:up-to-weighted}. The idea behind those algorithms is to
obtain rewriting rules from the pairs of $Z$ and two elements are in
the congruence closure iff they can be rewritten to the same normal
form.}

\paragraph*{Simulation equivalence.}
We illustrate how the infinite $(\alpha,\M$)-behavioural equivalence
game can be used to characterise simulation
equivalence~\cite{Glabbeek90} on serial LTS. We have described the
graded theory of simulation
in~\cref{E:semantics}.\ref{item:sem-sim}.  Recall that it requires
actions to be monotone, via the depth-1 equation
\( a(x + y) = a(x+ y) + a(x).  \) When trying to show that depth-1
terms $\sum_{i\in I}a_i(t_i)$ and $\sum_{j\in J}b_j(s_k)$ are \mbox{equal, D}
may exploit that over join semilattices, inequalities can be expressed
as equalities ($x\le y$ iff $x+y=y$), and instead endeavour to show
inequalities in both directions. By the monotonicity of actions,
$\sum_{i\in I}a_i(t_i)\le \sum_{j\in J}b_j(s_k)$ is implied by~D
claiming, for each~$i$, that $t_i\le s_j$ for some~$j$ such that
$a_i=b_j$; symmetrically for $\ge$ (and by the description of the
relevant graded monad as per
\cref{E:graded-theory}.\ref{item:simulation-a}, this proof
principle is complete). Once~S challenges either a claim of the form
$t_i\le s_j$ or one of the form $t_i\ge s_j$, the direction of
inequalities is fixed for the rest of the game; this corresponds to
the well-known phenomenon that in the standard pebble game for
similarity,~S cannot switch sides after the first move. Like for
bisimilarity (\cref{expl:bisim-instance}), the game can be modified
to let~S move first:~S first picks, say, one of the terms~$t_i$, and~D
responds with an~$s_j$ such that~$a_i=b_j$, for which she claims
$t_i\le s_j$. Overall, the game is played on positions in
$\pow^+(X)\times\pow^+(X)$, but if started on two states~$x,y$ of the
given labelled transition systems, i.e.~in a position of the form
$(\{x\},\{y\})$, the game forever remains in positions where both
components are singletons, and thus is effectively played on pairs of
states. Summing up, we recover exactly the usual pebble game for
mutual similarity. Variants such as complete, failure, or ready
simulation are captured by minor modifications of the graded
semantics~\cite{DMS19}.

\paragraph*{T-structured trace equivalence.}
Fix a set $\A$ and a finitary monad $T$ on $\Set$. We are going to
consider the $(\id, \M_T(\A))$-behavioural equivalence game on
coalgebras for the functor $T(\A\times -)$
(cf.~\cref{E:graded-monad}.\ref{item:T-traces}).

\begin{notn}
  Fix a presentation $(\Sigma', E')$ of $T$ (i.e.~an equational theory
  in the sense of universal algebra).  We generalize the graded trace
  theory described in \cref{E:graded-theory}.\ref{item:traces-a} to
  a graded theory $\T=(\Sigma,\E)$ for $\M_T(\A)$ as follows:
  $(\Sigma', E')$ forms the depth-0 part of $\T$ and, at depth-1, $\T$
  has unary actions $a(-)$ which distribute over all operations
  $f\in \Sigma'$:
\[
a(f(s_1,\cdots,s_{\arity(f)}))= f(a (s_1),\cdots,a(s_{\arity(f)}))
\]
The arising theory $\T$ presents $\M_T(\A)$.
\end{notn}

\noindent Recall from~\cref{R:coeq}.\ref{R:coeq:1} that since, in
this setting, $M_0=T$, a legal move for~D in position
$(s,t)\in TX\times TX$ is a relation $Z$ on $TX$ such
that equality of the respective successors~$(\alpha_X\cdot \gamma)^*_0(s)$
and $(\alpha_X\cdot \gamma)^*_0(t)$, viewed
as (equivalence classes of) depth-1 terms, is derivable in the theory
$\T$ under assumptions~$Z$.

\begin{rem}
  A natural question is whether there exist algorithms for deciding if
  a pair $(\alpha_X\cdot \gamma)^*_0(s),(\alpha_X\cdot \gamma)^*_0(t)$
  sits in the congruence closure of $Z$.  In fact, there are
  algorithms to check congruence closure of depth-0 terms for the
  powerset monad $T=\powf$~\cite{bp:checking-nfa-equiv} and for
  certain semiring monads~\cite{bkk:up-to-weighted}.  The idea behind
  those algorithms is to obtain rewrite rules from pairs in $Z$, and
  two elements are in the congruence closure if and only if they can
  be rewritten to the same \emph{normal form}. Applying depth-1
  equations to normal forms could potentially yield a method to check
  automatically whether a given pair of $M_1$-terms lies in the
  congruence closure of $Z$.
\end{rem}

\paragraph*{Finite-trace equivalence.}
More concretely, we examine the behavioural equivalence game for trace
equivalence on finitely branching LTS (i.e.
$(\id, \M_{\powf}(\A))$-semantics as
per~\cref{E:semantics}.\ref{item:sem-trace}).


\begin{expl}
  Consider the following process terms representing a coalgebra
  $\gamma$ (in a fragment of CCS):
  \[
  p_1\equiv a.p_1';
  \quad
  p_2\equiv a.p_2' + b.p_2'';
  \quad
  p_3\equiv b.p_3',
  \]
  where $p_1',p_2',p_2'',p_3'$ are deadlocked. It is easy to see that
  $s=\{p_1,p_2\}$ and $t=\{p_2,p_3\}$ are trace equivalent: In
  particular, $s,t$ have the same traces of length 1. We show that~D
  has a winning strategy in the 1-round
  $(\id, \M_{\powf}(\A))$-behavioural equivalence game at
  $(s,t)$. Indeed, the relation
  \( Z := \{p_1' + p_2' = p_2',\ p_3' + p_2'' = p_2'' \}, \)
  is admissible at~$(s, t)$: We must show that equality of
  $(\alpha\cdot\gamma)^\#(s)=a(p_1') + a(p_2') + b(p_2'')$
  and~$(\alpha\cdot\gamma)^\#(t)=a(p_2') + b(p_2'') + b(p_3')$ is
  entailed by~$Z$. To see this, note that
  \[
  Z\vdash a(p_1') + a(p_2') = a(p_2')\ \text{ and }\ Z\vdash b(p_2'') = b(p_3') + b(p_2'').
  \]
  Moreover, the pairs $(\{p_1',p_2'\},\{p_2'\})$
  and~$(\{p_3',p_2''\},\{p_2''\})$ are both identified by~$M_0!$ (all
  terms are mapped to $\{\epsilon\}$ when $1=\{\epsilon\}$).  That
  is,~$Z$ is a winning move for~D.
\end{expl}
\noindent
In general, admissible moves of~D can be described via a normalisation
of depth-1 terms as follows:

\begin{propn}\label{prop:trace-nf}
  In $\M_{\powf}(\A)$, every depth-1 term is derivably equal to one of
  the form $\sum_{a\in \A}a (t_{a})$, with depth-0 terms (i.e.~finite,
  possibly empty, sets)~$t_{a}$. Over serial LTS (i.e.~$T=\powf^+$),
  every depth-1 term has a normal form of the shape
  $\sum_{a\in B}a (t_{a})$ with $B\in\powf^+\A$ (where the~$t_a$ are
  now finite and non-empty).
\end{propn}
\begin{propn}\label{prop:WStratTrace}
  Let $\rho=\sum_{a\in\A} a(\rho_a)$ be depth-1 terms over~$X$ in
  normal form, for $\rho\in\{s,t\}$. Then a relation
  $Z\subseteq \powf X\times\powf X$
  is a legal move of D in position $(s,t)$ iff
  the following conditions hold for all $a\in\A$, in the notation
  of~\cref{prop:trace-nf}:
  \begin{enumerate}
    \item $\forall x\in s_{a}.\ \exists {t'}(t'\subseteq t_{a} \land Z\vdash x\leq t')$
    \item $\forall {y\in t_{a}}.\ \exists {s'}(s'\subseteq s_{a} \land Z\vdash y\leq s')$
  \end{enumerate}
  where, again, $s\leq t$ abbreviates $s+ t=t$. Over serial LTS
  (i.e.~$T=\powf^+$), and for normal forms
  $\rho=\sum_{a\in B_\rho} a(\rho_a)$, a relation
  $Z\subseteq \powf^+X \times \powf^+X$
  is a legal move of D in position~$(s,t)$
  iff~$B_s=B_t$ and the above conditions hold for all~$a\in B_s$.
\end{propn}
\noindent To explain terminology, we note at this point that by the
above, in particular $Z=\{(x,0)\mid x\in X\}\cup\{(0,y)\mid y\in X\}$
is always admissible. Playing~$Z$,~D is able to move in every
round, bluffing her way through the game; but this strategy does not
win in general, as her bluff is called at the end
(cf.~\cref{S:games}). More reasonable strategies work as follows.

On the one hand,~D can restrict herself to playing the bisimulation
relation on the determinised transition system because the term $s'$
(resp. $t'$) can be taken to be exactly $s_a$ (resp. $t_a$) in
Condition~2 (resp. Condition~1). This form of the game may be recast
as follows. Each round consists just of~S playing some~$a\in\A$ (or
$a\in B_s$ in the serial case), moving to $(s_a,t_a)$ regardless of
any choice by~D. In the non-serial case, the game runs until the bluff
is called after the last round. In the serial case, D wins if either
all rounds are played or as soon as~$B_s=B_t=\emptyset$, and~S wins as
soon as $B_s\neq B_t$.

On the other hand, D may choose to play in a more fine-grained manner,
playing one inequality $x\le t'$ for every $x\in s_a$ and one
inequality $s'\ge y$ for every $y\in t_a$. Like in the case of
simulation, the direction of inequalities remains fixed after~S
challenges one of them, and the game can be rearranged to let~S move
first, picking, say, $x\in s_a$ (or symmetrically), which~D answers
with $t'\subseteq t_a$, reaching the new position $x\le t'$. The game
thus proceeds like the simulation game, except that~D is allowed to
play sets of states.

\paragraph*{Probabilistic traces} These are treated similarly as traces in
non-deterministic LTS: Every depth-1 term can be normalized into one
of the form
 $ \sum_{\A} p_a\cdot a(t_a)$,
where $\sum_{\A}p_a=1$ and the~$t_a$ are depth-0 terms. To show
equality of two such normal forms $\sum_{a\in\A} p_a\cdot a(t_a)$ and
$\sum_{a\in\A} q_a\cdot a(s_a)$ (arising as successors of the current
configuration),~D needs to have $p_a=q_a$, and then claim~$t_a=s_a$,
for all $a\in\A$. Thus, the game can be rearranged to proceed like the
first version of the trace game described above:~S selects~$a\in\A$,
and wins if~$p_a\neq q_a$ (and the game then reaches the next
configuration~$(t_a,s_a)$ without intervention by~D).

\paragraph*{Failure equivalence.}
Let $\gamma\colon X \to \powf(\A \times X)$ be an LTS.
A tuple $(w,B)\in \A^* \times \powf(\A)$ is a \emph{failure pair} of a state $x$ if there is a $w$-path from $x$ to a state $x'\in X$
such that $x'$ \emph{fails} to perform some action $b\in B$ (the
\emph{failure set}). Two states are failure equivalent iff
they have the same set of failure pairs.

The \emph{graded theory of failure semantics}~\cite{DMS19}
extends the graded theory of traces by de\-pth-1 constants~$A$ for
each~$A\in\pow_f(\A)$ (failure sets) and depth-1
equations~$A+ (A\cup B) = A$ for each~$A,B\in\pow_f(\A)$
(failure sets are downwards closed). The resulting graded monad~\cite{DMS19}
has~$M_0X=\powf X$ and~${M_1X=\powf^\downarrow(\A\times X + \powf \A)}$,
where~$\powf \A$ is ordered by inclusion,~$\A\times X$ carries the discrete order, and
$\powf^\downarrow$ denotes the finitary downwards-closed powerset.
It is clear that $\barM$ still preserves monos since we have only expanded
the theory of traces by constants. The game in general
is then described similarly as the one for plain traces above; the key
difference is that now~S is able to challenge whether a pair of failure
sets are matched up to downwards closure.

\begin{expl}
  Consider the following process terms with $\A=\{a,b,c\}$:
  $
  p_1\equiv a.\dl,\ p_2\equiv a.\dl+b.\dl,\ p_3\equiv b.\dl.
  $

  Clearly, the states $s=\{p_1,p_2,p_3\}$ and $t=\{p_1,p_3\}$ in the pre-determinized system are failure equivalent. To reason this through our game $\CalG_\infty(\gamma)$, Duplicator starts with the relation $Z=\{(\dl,\dl)\}$. From~$Z$, we derive
  \begin{align*}
     (\alpha_X\cdot \gamma)^*_0 s &= a(\dl) + b(\dl) +  {\downarrow}\{b,c\} + {\downarrow}\{c\} + {\downarrow}\{a,c\} \\
    &=  a(\dl) + b(\dl) +   {\downarrow}\{b,c\} + {\downarrow}\{c\} + {\downarrow}\{c\} + {\downarrow}\{a,c\} \\
    &=  a(\dl) + b(\dl) + {\downarrow}\{b,c\} + {\downarrow}\{a,c\} =(\alpha_X\cdot \gamma)^*_0 t.
  \end{align*}
  Thus $Z$ is admissible at $(s,t)$ and the game position advances to $(\dl,\dl)$ from where Duplicator has a winning strategy.
\end{expl}

\section{Conclusions and Future Work}
\noindent We have shown how to extract characteristic games for a
given graded behavioural equivalence, such as similarity, trace
equivalence, or probablistic trace equivalence, from the underlying
graded monad, effectively letting Spoiler and Duplicator play out an
equational proof. The method requires only fairly mild assumptions on
the graded monad; specifically, the extension of the first level of
the graded monad to algebras for the zero-th level needs to preserve
monomorphisms. This condition is not completely for free but appears to
be unproblematic in typical application scenarios. In case the zero-th
level of the graded monad preserves the terminal object (i.e.~the
singleton set), it turns out that the induced graded behavioural
equivalence can be recast as standard coalgebraic behavioural
equivalence in a category of Eilenberg-Moore algebras, and is then
characterized by an infinite version of the generic equivalence
game. A promising direction for future work is to develop the generic
algorithmics and complexity theory of the infinite equivalence game,
which has computational content via the implied fixpoint
characterization. Moreover, we will extend the framework to cover
further notions of process comparison such as behavioural
preorders~\cite{FMS21a} and, via a graded version of quantitative
algebra~\cite{MPP16}, behavioural metrics.

\bibliography{monadsrefs}


\clearpage
\appendix
\section{Omitted Proofs and Details}

\subsection{Proofs for \cref{S:games}}

\begin{lem}\label{L:bar-kleisli}
Let~$\M$ be a depth-1 graded monad and
let~$X,Y$ be sets. Then for every function%
~$f\colon X\to M_kY$ we have\lsnote{@Chase: This should be a fact about general $f\colon X\to M_kY$ CF: Done.}
\[
\barM (f^*_0) = f^*_1.
\]
\end{lem}
\begin{proof}
    It follows by equational reasoning using the associative laws
    of $\M$ and naturality of $\mu^{0,n}$ that $f^*_n$
    is a homomorphism of $M_0$-algebras from~$(M_nX, \mu^{0,n})$
    to $(M_{n+m}X, \mu^{0, (n+m)})$. For the homomorphism%
    ~$f^*_0\colon M_0X\to M_kX$ there exists, by canonicity
    of~($M_0X, M_1X$) (see \cref{E:canonical}), a unique
    homomorphism~$h\colon M_1X\to M_{k+1}X$ making
    $(f^*_0, h)$ a homomorphism of $M_1$-algebras.
    Thus, it is enough to show that~$(f^*_0, f^*_1)$ is an
    $M_1$-algebra homomorphism.  Indeed, it will then follow
    that~$f^*_1$ is the $1$-part of the free extension of~$f^*_0$
    to an $M_0$-algebra homomorphism, i.e.~$\barM(f^*_0) = f^*_1$,
    as required.
    We conclude the proof by showing that
    $(f^*_0, f^*_1)$ satisfies the remaining interchange law stating that
    \[
      \mu^{1,k}_{M_{1+k}1}\cdot M_1(f^*_0) = f^*_1\cdot \mu^{1, 0}_X
    \]
    making it an $M_1$-algebra homomorphism:
    \begin{align*}
      \mu^{1,k}_{M_{1+k}1} \cdot M_1(f^*_0)
      &= (f_0^*)^*_1 & \text{def.~of $(-)^*_1$} \\
      &= (f_0^* \cdot \id_{M_0 X})^*_1 \\
      &= f_1^* \cdot (\id_{M_0 X})^*_1 & \text{by~\eqref{item:star-1}} \\
      &= f^*_1 \cdot \mu^{1,0}_X & \text{def.~of $(-)^*_1$}.
    \end{align*}
\end{proof}

\begin{expl}
Let~$(\alpha, \M)$ be a depth-1 graded semantics
on~$G$-coalgebras and let~$\gamma$ be a $G$-coalgebras.
We have
\begin{equation}
\barM(\gamma^{(k)})^*_0 = (\gamma^{(k)})^*_1		\tag{$k\in\omega$}
\end{equation}
by instantiation of~\cref{L:bar-kleisli} to the map
$\gamma^{(k)}\colon X\to M_kX$.
\end{expl}

\subsubsection*{Proof of~\cref{T:sound-complete}}
\begin{rem}\label{R:factorization}
  Whenever $c\colon X \epito C$ is the coequalizer of the kernel pair
  $p, q\colon Z\to X$ of a morphism $f\colon X\to Y$, then there is a
  unique morphism $m\colon C\to Y$ such that $m\cdot c = f$. Recall
  (e.g.~from the proof of Borceaux \cite[Thm.~2.1.3]{Borceux94}) that
  if $\CatC$ is a regular category, then~$m$ is monic; in particular,
  this holds when~$\CatC$ is the category of algebras for a monad
  on~$\Set$.
\end{rem}
\begin{proof}[Proof (\cref{T:sound-complete})]
  Let us denote the \emph{winning region} for D in $\CalG_n(\gamma)$ by
  $\Win_n(\text{D})$ (leaving the coalgebra $(X, \gamma)$ implicit in
  the notation).  We are going to show that $(s, t)\in\Win_n(D)$ if
  and only if $(\gamma^{(n)})^*_0$ merges $s$ and $t$ by induction on
  $n\in\omega$. For $n=0$, we must show that $M_0!(s)=M_0!(t)$ if and
  only if $(\gamma^{(0)})^*_0$ merges $s$ and $t$. In fact, we show
  that $M_0! = (\gamma^{(0)})^*_0$. First note that the following
  diagram commutes due to the unit law~\eqref{diagram:unitlaw} of $\M$
  (instantiated to $n=0$) and naturality of~$\mu^{0,0}$:
\[
\begin{tikzcd}[column sep = 35]
  M_0X \ar{rd}[swap,inner sep=0]{\id_{M_0X}} \arrow[r, "M_0\eta_X"]
  &
  M_0M_0X
  \arrow[r, "M_0M_0!"]
  \arrow[d, "{\mu^{0,0}_X}"']
  &
  M_0M_01
  \arrow[d, "{\mu^{0,0}_1}"]
  \\
  &
  M_0X \arrow[r, "M_0!"]
  &
  M_01
\end{tikzcd}
\]
Then unfold the definition of $(\gamma^{(0)})^*_0$ and compute
\begin{align*}
  (\gamma^{(0)})^*_0
  &= \mu^{0,0}_1\cdot M_0\gamma^{(0)}  \\
  &= \mu^{0,0}_1\cdot M_0(M_0!\cdot \eta_X)  \\		
  &= \mu^{0,0}_1\cdot M_0M_0!\cdot M_0\eta_X \\
  &= M_0!,
\end{align*}
where the last step uses that the outside of the previous diagram
commutes.

Now, inductively assume that $(s,t)\in\Win_{k}(D)$
if and only if $(\gamma^{(k)})^*_0$ merges $s$ and
$t$. We proceed to show that the equivalence carries
through to $n=k+1$:

\medskip
\noindent ($\Leftarrow$)~Let $(s,t)\in M_0X\times M_0X$
such that $(\gamma^{(k+1)})^*$ merges $s$ and~$t$ be given; we must
show that D has a winning strategy in the $(k+1)$-round
$\CalS$-behavioural equivalence game at $(s, t)$. Define
\[
Z:= \Ker((\gamma^{(k)})^*_0).
\]
That is, $Z$ is the pullback, computed in $\Set$, as
in the diagram
\[
\begin{tikzcd}[column sep = 35]
  Z \arrow[d, "\ell"', dashed] \arrow[r, "r", dashed]
  &
  M_0X \arrow[d, "(\gamma^{(k)})_0^*"]
  \\
  M_0X \arrow[r, "(\gamma^{(k)})_0^*"]
  &
  M_k1
\end{tikzcd}
\]
It is clear that (the relation induced by) $Z$ is a winning move for D
at $(s, t)$ provided that it is admissible.  Indeed, for every
$z\in Z$ we have that $(\gamma^{(k)})^*_0$ merges $\ell(z)$ and $r(z)$
by construction; so, by induction, $Z\subseteq\Win_{k}(D)$, as
claimed.  Thus, to conclude the proof of \emph{`if'}, it suffices to
show that $Z$ is admissible at $(s, t)$, i.e.~we want to show that the
homomorphism $\barZ\colon M_0 X \to \barM C_Z$ from~\eqref{eq:barZ}
merges $s$ and $t$.

Let $m\colon C_Z\to M_k 1$ be the unique (mono-)morphism such that
$m\cdot c_Z = (\gamma^{(k)})^*_0$ according
to~\cref{R:factorization}, and consider the following diagram of
$M_0$-algebra homomorphisms:
\[
\begin{tikzcd}[column sep = 40, row sep = 30]
  M_0X \arrow[rd, "(\gamma^{(k+1)})^*_0"'] \arrow[r,
  "(\alpha\cdot\gamma)^*_0"]
  \ar[shiftarr = {yshift=20}]{rr}{\overbar Z}
  &
  {M_1X=\barM M_0X}
  \ar{d}{\overbar M_1 (\gamma^{(k)})^*_0} \arrow[r, "\overbar M_1 c_Z"]
  &
  \barM C_Z \arrow[ld, "\overbar M_1 m"]
  \\
  &
  M_{k+1} 1 = \barM M_k1
\end{tikzcd}
\]
\takeout{
\[
\begin{tikzcd}[column sep = 35]
  M_0X
  \arrow[r, "(\alpha\cdot\gamma)^*_0"]
  \ar{rd}[swap]{(\gamma^{k+1})^*_0}
  \ar[shiftarr = {yshift=20}]{rr}{\barZ}
  &
  M_1X = \barM M_0X
  \arrow[r, "\barM c_Z"]
  \arrow[d, "\barM(\gamma^{(k)})^*_0"']
  \\
  \barM M_k1
  & \barM C_Z
  \arrow[rd, "\barM m"']
\end{tikzcd}
\]
}
The above diagram commutes: the upper part is the definition of
$\barZ$, the right hand triangle commutes since
$m\cdot c_Z = (\gamma^{(k)})^*_0$, and the left-hand triangle
commutes by the following:\cfnote{this needs to be included since it no longer follows from LemA.1}
\begin{align*}
\barM(\gamma^{(k)})^*_0\cdot (\alpha\cdot\gamma)^*_0
&= (\gamma^{(k)})^*_1\cdot (\alpha\cdot\gamma)^*_0			&\text{by~\cref{L:bar-kleisli}} \\
&= ((\gamma^{(k)})^*_1\cdot(\alpha\cdot\gamma))^*_0			&\text{by~\cref{item:star-1}} \\
&= (\gamma^{(k+1)})^*_0									&\text{~def. of~$\gamma^{(k+1)}.$}
\end{align*}
Thus
\[
  \barM m\cdot\barZ(s)
  =
  (\gamma^{(k+1)})^*_0(s)
  =
  (\gamma^{(k+1)})^*_0(t)
  =
  \barM m\cdot\barZ(t).
\]
Since $m$ is a monomorphism and $\barM$ preserves
monomorphisms, we conclude that $\barZ(s) = \barZ(t)$ holds.
Hence $Z$ is admissible at $(s, t)$, as desired.

\medskip
\noindent ($\Rightarrow$)~Suppose that $(s, t)\in\Win_{k+1}(\bbD)$ so
that there exists $Z\subseteq\Win_k(\bbD)$ which is admissible at
$(s, t)$; we proceed to show that $(\gamma^{(k+1)})^*_0$ merges~$s$
and~$t$, as required. By admissibility of~$Z$, we know that
$\barZ\colon M_0X \to \barM C_Z$ from~\eqref{eq:barZ}
merges~$s$ and~$t$.

Next we see that $(\gamma^{(k)})^*_0\colon M_0X\to M_k1$ coequalizes
the pair $\ell^*_0, r^*_0\colon M_Z\to M_0X$. Indeed, since
$Z\subseteq\Win_k(D)$, we have by induction that
\( (\gamma^{(k)})^*_0(\ell^*_0(z)) = (\gamma^{(k)})^*_0(r^*_0(z)) \)
for all $z\in Z$. Thus, by the universal property of the coequalizer
$c_Z\colon M_0X \to C_Z$ we obtain a unique homomorphism
$h\colon C_Z\to M_k1$ such that $h\cdot c_Z = (\gamma^{(k)})^*_0$.
Therefore the right-hand triangle of the following diagram of
$M_0$-algebra homomorphisms commutes:
\[
\begin{tikzcd}[column sep = 40, row sep = 30]
  M_0X \arrow[rd, "(\gamma^{(k+1)})^*_0"'] \arrow[r,
  "(\alpha\cdot\gamma)^*_0"]
  \ar[shiftarr = {yshift=20}]{rr}{\overbar Z}
  & M_1X = \barM M_0 X \arrow[d, "\overbar M_1 (\gamma^{(k)})^*_0"] \arrow[r, "\overbar M_1c_z"]
  & \barM C_Z \arrow[ld, "\overbar M_1 h", dashed] \\
  &
  M_{k+1}1 = \barM M_k 1
\end{tikzcd}
\]
Note that the upper part commutes by the definition~\eqref{eq:barZ} of
$\barZ$. The left-hand triangle commutes
by~\cref{L:bar-kleisli} as before. We conclude that
$(\gamma^{(k+1)})^*_0$ merges $s$ and $t$, as desired, since so does
$\barZ$.
\end{proof}

\subsection{Proofs for \cref{sec:infinte-depth}}
\subsubsection*{Details for \cref{R:final-coalg}}

We will prove that for every accessible graded monad $\M$ on a locally
presentable category, the category $\Alg_n(\M)$, $n\leq \omega$, is locally
presentable. Furthermore, we show that the functor $\barM\colon
\Alg_0(\M) \to \Alg_0(\M)$ is then accessible and therefore has a
final coalgebra.

Our proof makes use of several facts from the theory of accessible
and locally presentable categories which we tersely recall from
Ad\'amek and Rosick\'y's book~\cite{AR94}.

As before, we fix a regular cardinal $\kappa$. A diagram
\mbox{$D\colon \CatD \to \CatC$} is \emph{$\kappa$-filtered} if its diagram
scheme $\CatD$ is a \emph{$\kappa$-filtered} category, that is a
category in which every subcategory with less than $\kappa$ objects
and morphisms has a cocone. (In particular, $\CatD$ is then non-empty
since a cocone for the empty subcategory is some object of $\CatD$.) A
functor $F\colon \CatC \to \CatD$ is \emph{$\kappa$-accessible} if it
preserves $\kappa$-filtered colimits, and a monad is
$\kappa$-accessible if so is its underlying functor. An object $X$ of
a category $\CatC$ is \emph{$\kappa$-presentable} if its hom-functor
$\CatC(X,-)\colon \CatC \to \Set$ is $\kappa$-accessible.
\begin{rem}\label{R:adj}
  We shall need the following facts about $\kappa$-presentable
  objects and left adjoints.
  \begin{enumerate}
  \item\label{R:adj:1} A left adjoint $L\colon \CatC \to \CatD$ with a $\kappa$-accessible right adjoint
    $R$ preserves $\kappa$-presentable objects. Indeed, for every
    $\kappa$-presentable object $C$ in $\CatC$ and every
    $\kappa$-filtered diagram $D$ in~$\CatD$ we have the following chain
    of natural isomorphisms
    \begin{align*}
      \CatD(LC, \colim D) &\cong \CatC(C, R(\colim D)) \\
      &\cong \CatC(C, \colim RD) \\
      &\cong \colim \CatC(C, RD) \\
      &\cong \colim \CatD(LC, D).
    \end{align*}

  \item\label{R:adj:2} Recall (e.g.~\cite[Def.~7.74]{AHS90}) that an
    epimorphism $e$ is \emph{extremal} if whenever $e = m \cdot f$ for
    a monomorphism~$m$, then~$m$ is an isomorphism. Let
    $L \dashv R\colon \CatD \to \CatC$ be an adjunction where the
    right adjoint $R$ is faithful and reflects isomorphisms. Then the
    adjoint transpose of every extremal epimorphism $e\colon C \epito RD$
    is an extremal epimorphism, too.

    Indeed, to see that the adjoint transpose $\bar e\colon LC \to D$
    is epic take a pair of morphisms $f, g\colon D \to D'$ such that
    $f \cdot \bar e = g \cdot \bar e$. Then $Uf \cdot e = Ug \cdot e$
    by adjoint transposition, whence $Uf = Ug$ since $e$ is epic, and
    therefore $f= g$ since $U$ is faithful.

    Now suppose that $\bar e = m \cdot f$ for some monomorphism~$m$ in
    $\CatD$. Then $e = Um \cdot \bar f$ in $\CatC$, where $Um$ is
    monic since right adjoints preserve monos. Thus, $Um$ is an
    isomorphism in $\CatC$, whence $m$ is an isomorphism in $\CatD$
    since $U$ reflects isomorphisms.
  \end{enumerate}
\end{rem}

A category $\CatC$ is \emph{$\kappa$-accessible} if it has
$\kappa$-filtered colimits and a set of $\kappa$-presentable objects
such that every object of~$\CatC$ is a $\kappa$-filtered colimit of
object from that set. The category~$\CatC$ is locally
$\kappa$-presentable if it is locally $\kappa$-accessible and
cocomplete. We say that a category is \emph{accessible} (\emph{locally presentable})
if it is $\kappa$-accessible (locally $\kappa$-presentable) for some
regular cardinal $\kappa$.

\begin{rem}\label{R:acc}
  We now collect several facts used in the proof of \cref{T:acc} below.
  \begin{enumerate}
  \item\label{R:acc:1} A category is locally $\kappa$-presentable iff
    is is $\kappa$-accessible and complete~\cite[Cor.~2.47]{AR94}.


  \item\label{R:acc:1a} A \emph{strong generator} in a cocomplete
    category $\CatC$ is a set~$\CatG$ of objects such that every
    object is an extremal quotient of a coproduct of object from
    $\CatG$. A cocomplete category is locally $\kappa$-presentable iff
    it has a strong generator formed by~$\kappa$-presentable
    objects~\cite[Thm.~1.20]{AR94}.

  \item\label{R:acc:2} For a $\kappa$-accessible monad $T$ on a
    locally $\kappa$-presentable category $\CatC$, the category
    $\Alg(T)$ of all Eilenberg-Moore algebras is locally
    $\kappa$-presentable~\cite[Thm.~\& Rem.~2.78]{AR94}.

  \item\label{R:acc:2a} For every locally $\kappa$-presentable
    category $\CatC$ the product category $\CatC \times \CatC$ is
    locally $\kappa$-presentable. Indeed, first $\CatC \times \CatC$
    is $\kappa$-accessible~\cite[Prop.~2.67]{AR94}. In addition, the
    product is (co)complete since so is $\CatC$ and (co)limits are
    formed componentwise in the product category. Moreover, the product
    projection functors are clearly $\kappa$-accessible.

  \item\label{R:acc:3} Let $\CatC$ and $\CatK$ be locally
    $\kappa$-presentable, and let $F, G\colon \CatK \to \CatC$ be
    $\kappa$-accessible. The \emph{inserter category} of $F,G$ has
    objects $(K,f)$ where $K$ is an object of $\CatK$ and
    $f\colon FK \to GK$ is a morphism of $\CatC$; morphisms
    $h\colon (K,f) \to (K',f')$ are morphisms $h\colon K \to K'$ of
    $\CatK$ such that $f' \cdot Fk = Gk \cdot f$. The inserter
    category is $\lambda$-accessible for some regular cardinal
    $\lambda \geq \kappa$~\cite[Thm.~2.72]{AR94}. (Note that, in
    general, $\lambda \neq \kappa$.)%

  \item\label{R:acc:4} Let $\CatC$ and $\CatK$ be $\kappa$-accessible, and
    let $F, G\colon \CatK \to \CatC$ be $\kappa$-accessible functors.  The
    \emph{equifier} of a pair \mbox{$\alpha,\beta\colon F \to G$} of natural
    transformations is the full subcategory of $\K$ given by all
    objects $K$ such that $\alpha_K = \beta_K$. This category is
    $\kappa$-accessible and its inclusion functor into $\K$ is
    $\kappa$-accessible. The same holds, more generally,
    for any set of pairs $F^i, G^i\colon \CatK \to \CatC$ ($i \in I$)
    and any set of pairs
    $\alpha^i,\beta^i\colon F^i \to G^i$~\cite[Lem.~2.76]{AR94}.
  \end{enumerate}
\end{rem}

We say that a graded monad $\M$ on $\CatC$ is
\emph{$\kappa$-accessible} if every of its endofunctors $M_n$
($n \in \omega$) is $\kappa$-accessible.

\begin{thm}\label{T:acc}
  Let $\M$ be a $\kappa$-accessible graded monad on a locally
  $\kappa$-presentable category $\CatC$. Then the category
  $\Alg_n(\M)$ is locally $\kappa$-presentable for every
  $n \leq \omega$.
\end{thm}
\begin{rem}
  The proof is essentially a more involved variation on the proof of
  the result on categories $\Alg(T)$ of Eilenberg-Moore algebras
  mentioned in \cref{R:acc}.\ref{R:acc:2}. We provide the details
  for the convenience of the reader.
\end{rem}
\begin{proof}
  For $n = 0$ we are done since $\Alg_0(\M)$ is simply the
  Eilenberg-Moore category for the $\kappa$-accessible monad $M_0$. We
  now give an explicit proof for $n=1$; the general case is then an
  easy exercise.
  \begin{enumerate}
  \item Let $\CatK = \Alg_0(\M) \times \Alg_0(\M)$ be the product of
    the Eilenberg-Moore category for the monad $M_0$ with
    itself. Then~$\CatK$ is locally $\kappa$-presentable
    by~\cref{R:acc}.\ref{R:acc:2a}.

  \item Let $P_0, P_1\colon \CatK \to \CatC$ be the functors obtained
    by composing the product projections with the forgetful functor
    $U_0\colon \Alg_0(\M) \to \CatC$. By \cref{R:acc}.\ref{R:acc:3},
    the inserter category~$\CatI$ of the pair $M_1P_0, P_1$ of
    $\kappa$-accessible functors is $\lambda$-accessible for some
    regular cardinal $\lambda \geq \kappa$.

  \item Note that objects of the inserter category $\CatI$ are
    $5$-tuples $A = (A_0, A_1, (a^{i,j})_{i+j \leq 1})$, where
    $a^{i,j}\colon M_iA_j \to A_{i+j}$ is a morphism of $\CatC$
    (i.e.~the same data as objects of $\Alg_1(\M)$) such that
    $(A_0,a^{0,0})$ and $(A_1, a^{0,1})$ are Eilenberg-Moore algebras
    for the monad $M_0$ (the remaining axioms, homorphy and
    coequalization, of $M_1$-algebras need not hold). Let
    $V\colon \CatI \to \CatK$ be the functor forgetting the structure
    morphism~$a^{1,0}$.

  \item We prove that $V$ reflects isomorphisms. First note that
    morphisms in $\CatI$ from $(A_0, A_1, (a^{i,j})_{i+j \leq 1})$ to
    $(B_0, B_1, (b^{i,j})_{i+j \leq 1})$ are pairs $h_i\colon A_i \to
    B_i$, $i=0,1$, of $M_0$-algebra homomorphisms such that, in addition,
    we have $h_1 \cdot a^{1,0} = b^{1,0} \cdot M_1h_0$. Now given such
    a pair such that $V(h_0,h_1)$ is an isomorphism in $\CatK$, that
    means that $h_i$ is an isomorphism of $M_0$-algebras for $i = 0,1$, we need to show
    that the inverses $h_i^{-1}$ form a morphism in $\CatI$. To this
    end we make use of the following diagram:

    \[
      \begin{tikzcd}
        M_1B_0 \ar{r}{b^{1,0}} \ar{d}[swap]{M_1 h_0^{-1}}
        \ar[shiftarr = {xshift=-30}]{dd}[swap]{M_1\id}
        &
        B_1
        \ar{d}{h_1^{-1}}
        \ar[shiftarr = {xshift=20}]{dd}{\id}
        \\
        M_1A_0 \ar{r}{a^{1,0}} \ar{d}[swap]{M_1 h_0}
        &
        A_1
        \ar{d}{h_1}
        \\
        M_1B_0 \ar{r}{b^{1,0}}
        &
        B_1
      \end{tikzcd}
    \]
    Indeed, the outside, left- and right-hand parts and the lower
    square commute. Thus so does the desired upper square when
    postcomposed by the isomorphism $h_1$. Thus, this square
    commutes.

  \item We now prove that $V$ creates limits and $\kappa$-filtered
    colimits. Observe first that limits and $\kappa$-filtered colimits
    of $M_0$-algebras are formed at the level of their carrier objects
    since the forgetful functor $U_0\colon \Alg_0(\M) \to \CatC$
    creates limits and $\kappa$-filtered colimits (the latter since
    $M_0$ preserves them). Further note that all limits and colimits
    in $\CatK$ are formed componentwise.%
    \smnote{The following argument is just a variation of the standard
      argument that forgetful functors of categories algebras for a
      functor create all limits and all colimits that the functor preserves.}

    For limits, given a diagram of objects $A^j$ ($j\in J$) in
    $\CatI$, let~$A$ be its limit in $\CatK$ with the limit projections
    $p^j\colon A \to A^j$. Then we obtain a unique morphism
    $a^{1,0}\colon M_1 A_0 \to A_1$ such that every $p^j$ ($j \in J$)
    is a morphism in $\CatI$, i.e.~the following squares commute
    (throughout we put the index $j$ of the structure morphisms of
    $A^j$ in the subscript):
    \[
      \begin{tikzcd}
        M_1A_0
        \ar[dashed]{r}{a^{1,0}}
        \ar{d}[swap]{M_1 p^j_0}
        &
        A_1
        \ar{d}{p^j_1}
        \\
        M_1A^j_0
        \ar{r}{a^{1,0}_j}
        &
        A^j_1
      \end{tikzcd}
      \qquad
      \text{for every $j \in J$}.
    \]
    Indeed, the morphisms $a^{1,0}_j \cdot M_1p^j_0$ form a cone (in
    $\CatC$) on the diagram yielding the component $A_1$ of the
    colimit $A$: for every morphism $h = (h_0, h_1)\colon A^j \to A^k$
    in the given diagram scheme we have a commutative diagram
    \[
      \begin{tikzcd}[row sep = 3]
        &
        M_1A^j_0
        \ar{r}{a^{1,0}_j}
        \ar{dd}{M_1h_0}
        &
        A^j_1
        \ar{dd}{h_1}
        \\
        M_1A_0
        \ar{ru}[inner sep =0]{M_1p^j_0}
        \ar{rd}[swap, inner sep =0]{M_1p^k_0}
        \\
        &
        M_1A^k_0
        \ar{r}{a^{1,0}_k}
        &
        A^k_1
      \end{tikzcd}
    \]
    Thus, $A$ carries a unique structure of an $\CatI$-object such
    that the limit projections $p^j$ are morphisms of $\CatI$.

    We still need to prove that $A$ is the limit of the $A^j$ in
    $\CatI$. Given any cone $h^j\colon B \to A^j$ ($j \in J$) in
    $\CatI$, we know that there is a unique morphism $h\colon B \to A$
    of $\CatK$ such that $p^j \cdot h = h^j$ for every $j \in J$. It
    suffices to show that $h$ is a morphism in~$\CatI$. For this
    consider the following diagram:
    \[
      \begin{tikzcd}
        M_1B_0
        \ar{rrr}{b^{1,0}}
        \ar{ddd}[swap]{M_1h^j_0}
        &&&
        B_1
        \ar{ddd}{h^j_1}
        \\
        &
        M_1B_0
        \ar{r}{b^{1,0}}
        \ar{d}[swap]{M_1h_0}
        \ar{lu}[swap]{\id}
        &
        B_1
        \ar{d}{h_1}
        \ar{ru}{\id}
        \\
        &
        M_1A_0
        \ar{r}{a^{1,0}}
        \ar{ld}[swap,inner sep =0]{M_1p^j_0}
        &
        A_1
        \ar{rd}[inner sep =1]{p^j_1}
        \\
        M_1A^j_0
        \ar{rrr}{a^{1,0}_j}
        &&&
        A_1
      \end{tikzcd}
    \]
    The outside commutes since $h^j$ is a morphism of $\CatI$. The
    left- and right-hand parts commute by the definition of $h$, the
    lower part commutes by the definition of $a^{1,0}$ and the upper
    part trivially does. Thus, the middle square commutes when
    postcomposed by every limit projection $p^j_1$, whence this square
    commutes, as desired.

    For $\kappa$-filtered colimits, given a $\kappa$-filtered diagram
    of $M_1$-algebras $A^j$ ($j\in J$), let $A$ be its colimit in
    $\CatI$ with colimit injections $\inj^j\colon A^j \to A$. We use
    that $M_1$ preserves the colimit in the first component; that is,
    we have a $\kappa$-filtered colimit with the injections
    $M_1\inj^j_0\colon M_1A^j_0 \to M_1A_0$ ($j\in J$) in $\Alg_0(\M)$ (whence in
    $\CatC$). Then we obtain a unique algebra structure
    $a^{1,0}\colon M_1A_0 \to A_1$ such that every $\inj^j$
    ($j \in J$) is a morphism in $\CatI$, i.e.~the following squares
    commute:
    \[
      \begin{tikzcd}
        M_1A^j_0
        \ar{r}{a^{1,0}_j}
        \ar{d}[swap]{M_1\inj^j_0}
        &
        A^j_1
        \ar{d}{\inj^j_1}
        \\
        M_1A_0
        \ar[dashed]{r}{a^{1,0}}
        &
        A_1
      \end{tikzcd}
      \qquad
      \text{for every $j \in J$}.
    \]
    Indeed, the morphisms $\inj^j_1 \cdot a^{1,0}_j\colon M_1A^j_0 \to
    A_1$ form a cocone of the above diagram formed by the
    $M_0$-algebras $M_1A^j_0$: for every morphism $h = (h_0,h_1)\colon
    A^j \to A^k$ in the diagram scheme we have a commutative diagram
    \[
      \begin{tikzcd}[row sep = 3]
        M_1A^j_0
        \ar{r}{a^{1,0}_j}
        \ar{dd}[swap]{M_1h_0}
        &
        A^j_1
        \ar{rd}[inner sep = 0]{\inj^j_1}
        \ar{dd}[swap]{h_1}
        \\
        &&
        A_1
        \\
        M_1A^k_0
        \ar{r}{a^{1,0}_k}
        &
        A^k_1
        \ar{ru}[swap, inner sep =0]{\inj^k_1}
      \end{tikzcd}
    \]
    Thus, $A$ carries a unique structure of an $\CatI$-object such
    that the colimit injections $\inj^j$ are morphisms of $\CatI$.

    We still need to prove that $A$ is the colimit of the $A^j$
    in~$\CatI$. Given any cocone $h^j\colon A^j \to B$ ($j \in J$) in
    $\CatI$ we know that there is a unique morphism $h\colon A \to B$
    of $\CatK$ such that \mbox{$h \cdot \inj^j = h^j$} for every $j \in
    J$. It suffices to show that $h$ is a morphism in $\CatI$. For
    this consider the following diagram:
    \[
      \begin{tikzcd}
        M_1A^j_0
        \ar{rd}[inner sep=0]{M_1\inj^j_0}
        \ar{rrr}{a^{1,0}_j}
        \ar{ddd}[swap]{M_1h^j_0}
        &&&
        A^j_1
        \ar{ld}[swap,inner sep =0]{\inj^j_1}
        \ar{ddd}{h^j_1}
        \\
        &
        M_1A_0
        \ar{r}{a^{1,0}}
        \ar{d}[swap]{M_1h_0}
        &
        A_1
        \ar{d}{h_1}
        \\
        &
        M_1B_0
        \ar{r}{b^{1,0}}
        &
        B_1
        \\
        M_1B_0
        \ar{ru}{\id}
        \ar{rrr}{b^{1,0}}
        &&&
        B_1
        \ar{lu}[swap]{\id}
      \end{tikzcd}
    \]
    The outside commutes since $h^j$ is a morphism in $\CatI$. The
    left- and right-hand parts commmute by the definition of $h$, the
    upper part commutes by the definition of $a^{1,0}$ and the lower
    part trivially does. Thus, the inner square commutes when
    precomposed by every colimit injection $M_1\inj^j_0$, whence it
    commutes, as desired.

    \takeout{
      One readily proved that $V$ creates limits and $\kappa$-filtered
      colimits. The proof is a variation on the standard argument that
      the forgetful functor of a category of algebras for an
      endofunctor creates all limits and all those colimits that the
      functor preserves. Here the argument needs to be adjusted to
      accomodate the structure morphisms
      $a^{1,0}\colon M_1A_0 \to A_1$ of objects of $\CatI$ and making
      use of the $\kappa$-accessibility of $M_1$. One also uses that
      limits and $\kappa$-filtered colimits of Eilenberg-Moore
      algebras for $M_0$ are formed at the level of $\CatC$, that is,
      the forgetful functor $U_0\colon \Alg_0(\M) \to \CatC$ creates
      limits and $\kappa$-filtered colimits. We leave the
      straightforward details to the reader.}

  \item It follows that $\CatI$ is complete, and therefore it is
    locally presentable by \cref{R:acc}.\ref{R:acc:1}. Moreover, by
    the Adjoint Functor Theorem~\cite[Thm.~1.66]{AR94}, $V$ has a left
    adjoint \mbox{$L\colon \CatK \to \CatI$}.

    From \cref{R:adj}.\ref{R:adj:1} we obtain that for every
    $\kappa$-presentable object in $\CatK$, i.e.~a pair $A$ formed by
    $\kappa$-presentable objects $(A_i, a^{0,i})$, $i = 0,1$, the
    object $LA$ is $\kappa$-presentable in $\CatI$. We show that the
    these $\kappa$-presentable objects $LA$ form a strong generator in
    $\CatI$. By \cref{R:acc}.\ref{R:acc:1a}, it follows that $\CatI$
    is locally $\kappa$-presentable. Given an object $A$ of $\CatI$
    we know that $VA$ is an extremal quotient of a coproduct of
    $\kappa$-presentable objects in $\CatK$, reprensented, say, by the
    extremal epimorphism
    \[\textstyle
      e\colon \coprod_{j\in J} A_j \epito VA
    \]
    for some set $J$. Using \cref{R:adj}.\ref{R:adj:1}, that the
    left-adjoint $L$ preserves coproducts and that $V$ reflects
    isomorphisms and is clearly faithful we thus see that the adjoint
    transpose
    \[\textstyle
      \bar e\colon \coprod_{j\in J} LA_j \epito A
    \]
    exhibits $A$ as an extremal quotient of the coproduct of the
    $\kappa$-presentable objects $LA_j$.

  \item Now let $Q_i\colon \CatI \to \CatC$, $i = 0,1$, be the obvious
    forgetful projection functors. We have natural transformations
    $\varphi^{i,j}\colon M_iQ_j \to Q_{i+j}$ for $i+j \leq 1$ with
    components $\varphi^{i,j}_A = a^{i,j}$. The missing laws of an
    $M_1$-algebra (homomorphy and coequalization) can be expressed by
    equifiers of the following pairs of natural transformations
    \begin{align*}
      \varphi^{1,0} \cdot \mu^{1,0} Q_0, \varphi^{1,0} \cdot
      M_1\varphi^{0,0}
      &\colon M_1M_0Q_0 \to Q_1,
      \\
      \varphi^{1,0}\cdot \mu^{1,0} Q_0, \varphi^{1,0}\cdot
      M_0\varphi^{1,0}
      &
      \colon M_0M_1Q_0 \to Q_1.
    \end{align*}
    Consequently, $\Alg_1(\M)$ is $\kappa$-accessible and the
    inclusion functor $I\colon\Alg_1(\M) \subto \CatI$ is $\kappa$-accessible
    by \cref{R:acc}.\ref{R:acc:4}.

  \item To see that $\Alg_1(\M)$ is locall $\kappa$-presentable, we
    prove that $\Alg_1(\M)$ is closed under limits and
    $\kappa$-filtered colimits in $\CatI$. This is a slight variation
    of the standard argument that the Eilenberg-Moore
    algebras for a $\kappa$-accessible monad $T$ is closed under
    limits and $\kappa$-filtered colimits.

    For limits, given a diagram of
    $M_1$-algebras $A^j$ ($j \in J$), let $A$ be its limit in $\CatI$
    with limit projections $p^j\colon A \to A^j$. We show that $A$ is an
    $M_1$-algebra, that means that it satisfies homomorphy and
    coequalization. For the latter consider the following diagram
    (structure morphisms of $A^j$ have their index as a subscript,
    e.g.~$a^{1,0}_j\colon M_1A^j_0 \to A^j_1$):
    \[
      \begin{tikzcd}[row sep=15]
        M_1M_0A^j_0
        \ar{rrr}{\mu^{1,0}_{A^j_0}}
        \ar{ddd}[swap]{M_1a^{0,0}_j}
        &&&
        M_1A^j_0
        \ar{ddd}{a^{1,0}_j}
        \\
        &
        M_1M_0A_0\ar{lu}[swap,inner sep=0]{M_1M_0p_0^j}
        \ar{r}{\mu^{1,0}_{A_0}}
        \ar{d}[swap]{M_1a^{0,0}}
        &
        M_1A_0
        \ar{ru}[inner sep=0]{M_1p_0^j}
        \ar{d}{a^{1,0}}
        \\
        &
        M_1A_0
        \ar{ld}[inner sep = 0]{M_1p_0^j}
        \ar{r}{a^{1,0}}
        &
        A_1
        \ar{rd}[swap, inner sep = 1]{p_1^j}
        \\
        M_1A^j_0
        \ar{rrr}{a^{1,0}_j}
        &&&
        A^j_1
      \end{tikzcd}
    \]
    The outside commutes due to coequalization of $A^j$ and the four
    inner trapezoids by the naturality of $\mu^{1,0}$ and since
    $p^j = (p^j_0, p^j_1)$ is a morphism in $\CatI$. Thus, the desired
    middle square commutes when postcomposed by the projection
    $p^j_1$. Since limits in $\CatI$ are formed componentwise, the
    $p^j_1$ ($j \in J$) form a limit cone whence a jointly monic
    family. This, implies that the desired inner square commutes.
    That $A$ satisfies homomorphy is shown analogogously.

    For $\kappa$-filtered colimits, given a $\kappa$-filtered diagram of
    $M_1$-algebras $A^j$ ($j \in J$), let $A$ be its colimit in $\CatI$
    with colimit injections $\inj^j\colon A^j \to A$. We show that $A$ is an
    $M_1$-algebra, that means that it satisfies homomorphy and
    coequalization. For the former, observe first that since $M_0M_1$
    is $\kappa$-accessible and colimits in $\CatI$ are formed
    componentwise we have a colimit cocone $M_0M_1\inj^j_0\colon A^j_0
    \to A_0$. Now consider the following diagram (again we put the
    index of $A^j$ in the subscript of structure morphisms)
    \[
      \begin{tikzcd}[row sep=15]
        M_0M_1A_0^j
        \ar{rrr}{\mu^{0,1}_{A_0^j}}
        \ar{ddd}[swap]{M_0a^{1,0}_j}
        \ar{rd}[inner sep = 0]{M_0M_1\inj^j_0}
        &&&
        M_1A^j_0
        \ar{ddd}{a^{1,0}_j}
        \ar{ld}[swap, inner sep = 1]{M_1\inj^j_0}
        \\
        &
        M_0M_1A_0
        \ar{r}{\mu^{0,1}_A}
        \ar{d}[swap]{M_0a^{1,0}}
        &
        M_1A
        \ar{d}{a^{1,0}}
        \\
        &
        M_1A_0
        \ar{r}{a^{1,0}}
        &
        A_1
        \\
        M_1A^j_0
        \ar{rrr}{a^{1,0}_j}
        \ar{ru}[inner sep = 1]{M_1\inj^j_0}
        &&&
        A_1^j
        \ar{lu}[swap, inner sep = 0]{\inj^j_1}
      \end{tikzcd}
    \]
    The outside commutes since $A^j$ is an $M_1$-algebra, and the four
    inner trapezoids commute due to the naturality of $\mu^{0,1}$ and
    the fact that $\inj^j$ is a morphism in $\CatI$. Thus, the desired
    inner square commutes when precomposed by every colimit injection
    $M_0M_1\inj^j_0$, whence it commutes as desired. That $A$
    satisfies coequalization is shown analogously

  \item The desired result now follows from the Reflection Theorem~\cite[Thm.~2.48]{AR94}.
    \qedhere
  \end{enumerate}
\end{proof}
\begin{corollary}
  The functor $(-)_1\colon \Alg_1(\M) \to \Alg_0(\M)$ taking $1$-parts
  is $\kappa$-accessible.
\end{corollary}
\noindent
Indeed, this follows from the proof of \cref{T:acc} since we have
the commutative diagram of functors
\[
  \begin{tikzcd}[row sep = 3]
    \Alg_1(\M) \ar{rr}{(-)_1}
    \ar[hook]{rd}[swap]{I}
    &&
    \Alg_0(\M)
    \ar{dd}{U_0}
    \\
    &
    \CatI
    \ar{rd}[swap]{Q_1}
    \\
    &&
    \CatC
  \end{tikzcd}
\]
where the functor $Q_1I$ preserves and the forgetful functor~$U_0$
creates $\kappa$-filtered colimits.
\begin{corollary}
  The functor $\barM\colon \Alg_0(\M) \to \Alg_0(\M)$ has a final
  coalgebra.
\end{corollary}
\noindent
Indeed, since $\barM = (-)_1 \cdot E$, where $E$ is a left adjoint
(see~\eqref{eq:barM}) we see that $\barM$ is a $\kappa$-accessible
endofunctor on the locally $\kappa$-presentable category
$\Alg_0(\M)$. It follows that the category of all coalgebra for
$\barM$ is locally presentable~\cite[Exercise~2j]{AR94}. (Though note
that the index of presentability may not be $\kappa$.) In particular,
this category is complete and therefore has a terminal object.

\subsubsection*{Proof of \cref{thm:infinite-depth-games}}

\begin{rem}\label{R:fact}
  Suppose that $\CatC$ is a category with a (regular epi, mono)
  factorization system (e.g.~a regular
  category~\cite[Def.~2.1.1]{Borceux94}); that means that every
  morphism $f$ of $\CatC$ has a factorization $f = m\cdot e$ into a
  regular epimorphism $e$ followed by a monomorphism $m$. Moreover,
  the following unique diagonal fill-in property holds: whenever we
  have $m \cdot f = g \cdot e$ for a regular epimorphism $e$ and a
  monomorphism $m$, then there exists a unique diagonal $d$ such that
  $m \cdot d = g$ and $d \cdot e = f$:
  \[
    \begin{tikzcd}
      A \ar[->>]{r}{e}
      \ar{d}[swap]{f}
      &
      B
      \ar{d}{g}
      \ar[dashed]{ld}[swap]{d}
      \\
      C
      \ar[>->]{r}{m}
      &
      D
    \end{tikzcd}
  \]

  Now suppose further that $F\colon \CatC \to\CatC$ preserves
  monomorphisms. Then every coalgebra morphism
  $h\colon (A,\gamma) \to (B,\delta)$ can be factorized into a
  coalgebra morphism carried by a regular epi and one carried by a
  monomorphism in $\CatC$. Indeed, take the factorization $h$ into a
  regular epimorphism $e\colon A \epito C$ followed by a monomorphism
  $m\colon C \monoto B$ in $\CatC$, and observe that the unique
  diagonal fill-in property yields a unique coalgebra structure
  $\alpha\colon C \to FC$ such that $e$ and $m$ are coalgebra
  morphisms (here one uses that $Fm$ is a monomorphism):
  \[
    \begin{tikzcd}
      A \ar{r}{\gamma}
      \ar[->>]{d}[swap]{e}
      \ar[shiftarr = {xshift=-18}]{dd}[swap]{h}
      &
      FA
      \ar{d}{Fe}
      \ar[shiftarr = {xshift=20}]{dd}{Fh}
      \\
      C
      \ar[dashed]{r}{\alpha}
      \ar[>->]{d}[swap]{m}
      &
      FC
      \ar[>->]{d}{Fm}
      \\
      B
      \ar{r}{\delta}
      &
      FB
    \end{tikzcd}
  \]
\end{rem}

Towards the proof of \cref{thm:infinite-depth-games}, we generalize
the game to be played on any $\barM$-coalgebra $(A,\gamma)$. Positions
of~D then are pairs in $A\times A$; positions of~S are relations
$Z\subseteq A\times A$. A move from $(x,y)\in A \times A$ to such
a~$Z$ is admissible for~D if the relations in~$Z$ entail equality
of~$\gamma(x)$ and~$\gamma(y)$; in categorical formulation, this means
that $\barM c_Z(\gamma(x))=\barM c_Z(\gamma(y))$ where
$c_Z\colon A\to C_Z$ is the coequalizer of the pair
$\ell^*_0,r^*_0\colon M_0Z\to A$ for the projections
$\ell,r\colon Z\to A$. Again,~S just picks from~$Z$; any player who
cannot move, loses, and infinite matches are won by~D. We then claim,
generalizing the claim of the theorem, that $x,y\in A$ are
behaviourally equivalent iff~D wins $(x,y)$.

\begin{proof}[Proof (\cref{thm:infinite-depth-games})]
($\Leftarrow$)~We denote the winning region of~$D$ by $\infgameeq$. It is
easy to see that~$\infgameeq$ is a congruence on the $M_0$-algebra~$A$
(this is similar as in the proof of \cref{thm:fin-coalg}),%
\smnote{I do not see this; can you please insert more details.}
so
we have a unique $M_0$-algebra structure on $A/{\infgameeq}$ such that the
quotient map $e\colon A\epito A/{\infgameeq}$ is an $M_0$-algebra
homomorphism. We write elements of $A/{\infgameeq}$ in the form $[x]_D$
for $x\in A$. We shall define a coalgebra structure $\gamma_D\colon
A/{\infgameeq} \to \barM(A/{\infgameeq})$ such that the square below commutes
in $\Alg_0(\M)$:
\[
  \begin{tikzcd}
    A \ar{r}{\gamma}
    \ar[->>]{d}[swap]{e}
    &
    \barM A
    \ar{d}{\overbar M_1 e}
    \\
    A/{\infgameeq}
    \ar[dashed]{r}{\gamma_D}
    &
    \barM(A/{\infgameeq})
  \end{tikzcd}
\]
That means that it suffices to show that
\begin{equation*}
  \gamma_D([x]_D) = \barM e(\gamma(x))\qquad
  \text{for every $x \in A$}
\end{equation*}
is well-defined; indeed, since $\barM e\cdot \gamma$ is an $M_0$-algebra
homomorphism and $\infgameeq$ a congruence, $\gamma$ is then clearly an
$M_0$-algebra homomorphism, too, and the above diagram commutes so
that $e$ becomes a coalgebra morphism that identifies all $x,y$ such
that~D wins~$(x,y)$).

So let~D win $(x,y)$; we have to show that
\[
  \barM e(\gamma(x))=\barM e(\gamma(y)).
\]
Let~$Z$ be the winning move of~D at $(x,y)$ (that is,~$Z$ is D's first
move of the corresponding match). Since~S can pick any element
of~$Z$,~D wins on every element of~$Z$, so that~$e$ factors through
the above coequalizer~$c_Z\colon A \to C_Z$: there is a unique
$M_0$-algebra homomorphism $h\colon C_Z \to A/{\infgameeq}$ such that
$h \cdot c_Z = e$. Hence,~$\barM e$ factors through~$\barM c_Z$ via
$\barM h$. Since~$Z$ is an admissible move, $\barM c_Z$
identifies~$\gamma(x)$ and~$\gamma(y)$; hence, so does~$\barM e$, as
required.%
\smnote{This is all categorical, so why don't we make a categorical
  proof?}

\medskip
\noindent
($\Rightarrow$)~Let $h\colon(A,\gamma)\to(B,\delta)$ be an
$\barM$-coalgebra morphism; it suffices to show that the kernel
\[
  \Ker h = \{(x,y)\in A\times A \mid h(x)=h(y)\}
\]
of~$h$ is contained in~$\infgameeq$. To see this, it suffices to show
that~D can maintain the invariant $\Ker h$, i.e.~ensure that her
positions always remain in $\Ker h$ if the game starts in a position
in $\Ker h$. But when at a position~$(x,y)\in\Ker h$, she can clearly
ensure that the next position is still in $\Ker h$ by just playing
$Z :=\Ker h$.

We proceed to show that $Z$ is admissible. Let
$\ell, r\colon Z \to A$ be the obvious projection maps. Since the
kernel is clearly a congruence,~$Z$ is an $M_0$-algebra, and
$\ell, r\colon Z \to A$ are $M_0$-algebra homomorphims. (In fact,~$Z$
is the pullback of~$h$ along itself, and the forgetful functor
$\Alg_0(\M) \to \Set$ creates limits.)  Using
\cref{R:coeq}.\ref{R:coeq:2} (and noting that it holds for
arbitrary algebras $A$ in lieu of $M_0X$) we take the coequalizer
$c_Z\colon A \epito C_Z$ of the pair $\ell,r$, and we now prove that
$\barM c_Z \cdot \gamma $ merges the pair $\ell, r$. 
By \cref{R:factorization}, the unique morphism
$m\colon C_Z \monoto B$ such that $h = m \cdot c_Z$ is monic. By
\cref{R:fact} we know that $C_Z$ carries a unique coalgebra
structure $\alpha\colon C_Z \to \barM C_Z$ such that $c_Z$ and $m$ are
coalgebra morphisms. Using the former fact and that $c_Z$ merges
$\ell$ and $r$ we obtain the following commutative diagram:
\[
  \begin{tikzcd}
    &
    A
    \ar{r}{\gamma}
    \ar[->>]{d}{c_Z}
    &
    FA
    \ar{d}{\overbar M_1 c_Z}
    \\
    Z
    \ar{ru}{\ell}
    \ar{rd}[swap]{r}
    &
    C
    \ar{r}{\alpha}
    &
    FC
    \\
    &
    A\ar{r}{\gamma}
    \ar[->>]{u}[swap]{c_Z}
    &
    FA
    \ar{u}[swap]{\overbar M_1 c_Z}
  \end{tikzcd}
\]
The commutativity of its outside is the desired equality.%
\takeout{
  Further we consider the kernel is clearly a congruence,~$Z$ is an
  $M_0$-algebra, and $\ell, r\colon Z \to A$ are $M_0$-algebra
  homomorphims. (In fact,~$Z$ is the pullback of~$h$ along itself, and
  the the forgetful functor $\Alg(M_0) \to \Set$ creates limits.)
  Using \cref{R:coeq}.\ref{R:coeq:2} we take the coequalizer of
  $c_Z\colon A \epito C_Z$ and show that $\barM c_Z \cdot \gamma$
  merges the pair $\ell, r$. We now use the well-known fact that the
  algebras for a monad (here: $M_0$) form a regular category. Hence,
  by \cref{R:factorization} we have a unique monomorphism
  $m\colon C_Z \monoto B$ such that $h = m \cdot c_Z$. Further we
  consider the following diagram%
\[
  \begin{tikzcd}
    &&
    Z
    \pullbackangle{-90}
    \ar{ld}[swap]{\ell}
    \ar{rd}{r}
    \\
    &
    A
    \ar[->>]{r}{c_Z}
    \ar{rd}[swap]{h}
    \ar{ld}[swap]{\gamma}
    &
    C_Z
    \ar[>->]{d}{m}
    &
    A
    \ar[->>]{l}[swap]{c_Z}
    \ar{ld}{h}
    \ar{rd}{\gamma}
    \\
    \barM A
    \ar{d}[swap]{\overbar M_1 c_Z}
    \ar{rrd}{\overbar M_1 h}
    &&
    B
    \ar{d}{\beta}
    &&
    \barM A
    \ar{d}{\overbar M_1 c_Z}
    \ar{lld}[swap]{\overbar M_1 h}
    \\
    \barM C_Z
    \ar[>->]{rr}{\overbar M_1 m}
    &&
    \barM B
    &&
    \barM C_Z
    \ar[>->]{ll}[swap]{\overbar M_1 m}
  \end{tikzcd}
\]
All its inner parts commute; the middle left- and right-hand
trapezoids commute because $h$ is a coalgebra morphism. Thus the outside
of the diagram commutes. Since $\barM$ preserves monomorphisms we have
that $\barM m$ is a monomorphism. Thus $\barM c_Z \cdot \gamma$ merges
$\ell, r$ as desired.}
\takeout{
Since~$\barM$ preserves
monomorphisms, $\barM h$ also factors through $\barM c_Z$ via the
monomorphism $\barM m$. It thus suffices to show that $\barM h$
identifies~$\gamma(x)$ and~$\gamma(y)$, which we show using that~$h$
is a coalgebra morphism and that $(x,y)\in\Ker h$:
\begin{equation*}
  \barM h(\gamma(x)) = \beta(h(x))=\beta(h(y))=\barM h(\gamma(y)).\tag*{\qed}
\end{equation*}}
\end{proof}

\subsubsection*{Proof of \cref{thm:fin-coalg}}
It is straightforward to see that~$\infgameeq$ is an equivalence
relation. For instance, transitivity is seen as follows. Assume that~D
wins on $(s,t)$ and on $(t,u)$, with winning (first) moves~$Z,Z'$,
respectively. Then $Z\cup Z'$ is a winning move for~D on $(s,u)$,
where we exploit that by the assumption that $\barM$ preserves monos,
we do not need to worry about $Z\cup Z'$ possibly using more variables
than needed in the top-level decompositions of~$s$ and~$u$
(cf.~\cref{rem:monos}). Symmetry and reflexivity are easier. We
write $[t]^D=q(t)$.  The $M_0$-algebra structure of~$U$ is then given
by
  \begin{equation*}
    f^U(([t_i]^D)_{i\in I})=[f((t_i)_{i\in I})]^D
  \end{equation*}
  for a depth-0 operation~$f$ of arity~$I$. Well-definedness is seen
  similarly as transitivity above.

  For well-definedness of~$\zeta$, suppose that $t_1\sigma, s_1\sigma$
  are top-level decompositions of infinite-depth terms such that~D
  wins on~$(t_1\sigma,s_1\sigma)$, again w.l.o.g.~with
  $t_1,s_1\in\Termsarg{\Sigma,1}(X)$ for the same~$X$, and with the
  same~$\sigma$.
  Let~$Z\subseteq\Termsarg{\Sigma,1}(X)\times\Termsarg{\Sigma,1}(X)$
  be~D's winning move. Then~$Z$, being an admissible move of~D,
  entails $t_1=s_1$. Moreover, since~S can pick any element
  $(u,v)\in Z$ as a response,~D wins on each of the subsequent
  positions $(u\sigma,v\sigma)$, so~$u$ and~$v$ are identified under
  $q\cdot\sigma$. It follows that
  $M_1(q\cdot\sigma)([t_1])=M_1(q\cdot\sigma)([s_1])$. This shows
  well-definedness; it is then clear by construction that~$\zeta$ is
  an~$M_0$-homomorphism.%
  \smnote{I am lost here. Can this be demonstated?}

  Finally, let $\gamma\colon A\to\barM A$ be an $\barM$-coalgebra on
  an $M_0$-coalgebra~$A$; we abuse~$A$ to denote also the carrier
  of~$A$.%
  \smnote{This should not be said here, but be a global convention
    stated in prelims.}
  For each~$x\in A$, the successor structure
  $\gamma(x)\in\barM A$ is represented as a term
  $g_x\in\Termsarg{\Sigma,1}(A)$. Infinite unrolling of~$A$ thus
  produces an infinite-depth term $h_x$ for each~$x\in A$. We claim
  that the map~$h\colon A\to U$ given by $h(x)=[h_x]^D$ is the unique
  $\barM$-coalgebra morphism from $(A,\gamma)$ to $(U,\zeta)$.

  First note that the map~$h$ is independent of the choice of
  representing terms~$g_x$. This is seen by letting~D play, in
  $\CalG_\infty^\syn$, a strategy based on maintaining the invariant
  that the present state is a pair of unrollings, under differently
  chosen representing terms (where the choice of representing terms may change
  during the unrolling process), of the same state in~$A$.
  Like in the
  `only if' direction of the proof of
  \cref{thm:infinite-depth-games},~D can maintain the invariant by
  playing it in every move.%
  \smnote{I do not understand the meaning here: `it' $=$ `the
    invariant'; so why can~D play this?}  To see that~$h$ is an
  $M_0$-homomorphism, let~$f$ be a depth-0 operation of arity~$I$;
  then $\gamma(f(x_i)_{i\in I})=f(\gamma(x_i)_{i\in I})$ has
  $f((g_{x_i})_{i\in I})$ as a representing term, so the unrolling
  $h(f(x_i)_{i\in I})$ arises by applying~$f$ to the
  unrollings~$h(x_i)$.

  Finally, being a coalgebra morphism $A\to U$ amounts precisely to
  the unrolling property, so~$h$ is a coalgebra morphism, and unique
  as such.
  \qed

\subsection*{Proof of Proposition~\ref{prop:WStratTrace}}
\begin{proof}
  We first show that $Z$ is admissible at $(s,t)$,
  i.e.~$Z \vdash (\alpha_X\cdot\gamma)^*_0 s=
  (\alpha_X\cdot\gamma)^*_0t$, where we write $Z\vdash$ to indicate
  equational derivability from~$Z$. This follows from the construction
  of $\bar Z$ since Condition 1 ensures that
  \[
  Z \vdash s_{a} + t_{a} = \sum_{x\in s_{a}} x + t_{a} = t_{a}.
  \]
  Likewise Condition 2 ensures that $Z \vdash s_a + t_a = s_a$. Thus, $Z \vdash s_a=t_a$; whence, $Z \vdash (\alpha_X\cdot\gamma)^*_0s=(\alpha_X\cdot\gamma)^*_0 t$.

  Moreover, Spoiler thanks to Condition~1 (resp.~Condition~2) can play
  the positions either of the form $(x\lor t',t')$ or $(y\lor
  s',s')$. In either cases, $M_0!$ will map the left and right term to
  $1$ since $M_01=1$. So, in hindsight, if $(s,t) Z (s',t')$ is a
  winning match for Duplicator, then $s'\mathrel Z t'$. Thus,
  repeating the above argument $n$-times give the desired
  result. Lastly, $Z$ is also a winning strategy in the infinite game
  because $(s',t')\in Z$ whenever $(s,t) Z (s',t')$ is a match.
\end{proof}


\end{document}